\documentclass[a4paper]{amsart}
\usepackage{a4wide}
\usepackage{mathrsfs}
\usepackage[colorlinks=true,linkcolor=blue,citecolor=red]{hyperref}
\usepackage{amsmath,amsthm,amssymb,enumitem}
\usepackage{multicol}

\usepackage[colorinlistoftodos]{todonotes}
\usepackage{soul} 

\newtheorem{theorem}{Theorem}[section]
\newtheorem{proposition}[theorem]{Proposition}
\newtheorem{lemma}[theorem]{Lemma}
\newtheorem{corollary}[theorem]{Corollary}
\newtheorem{conjecture}[theorem]{Conjecture}
\newtheorem{question}[theorem]{Question}

{
\theoremstyle{definition}
\newtheorem{definition}[theorem]{Definition}
}
\newcommand{\N}{\mathbb{N}}

\newcommand{\sA}{\mathbb{A}}
\newcommand{\sB}{\mathbb{B}}
\newcommand{\sC}{\mathbb{C}}
\newcommand{\sD}{\mathbb{D}}
\newcommand{\sF}{\mathbb{F}}
\newcommand{\sG}{\mathbb{G}}

\newcommand{\sH}{\mathbb{H}}
\newcommand{\sS}{\mathbb{S}}
\newcommand{\sX}{\mathbb{X}}
\newcommand{\sY}{\mathbb{Y}}

\newcommand{\EXPTIME}{{\sc ExpTime}}

\DeclareMathOperator\Aut{Aut}
\DeclareMathOperator\Pol{Pol}

\DeclareMathOperator\CSS{CSS}
\DeclareMathOperator\csp{CSP}

\DeclareMathOperator\ar{ar}

\newcommand\PSPACE{\textsc{Pspace}}
\newcommand\coNP{\ensuremath{\mathrm{co}\textsc{NP}}}
\newcommand\NP{\textsc{NP}}

\newcommand\dwnu{dissected weak near-unanimity}

\theoremstyle{remark}
\newtheorem*{remark}{Remark}

\newcommand{\Hrushovski}[1]{\Hru{#1}}
\newcommand{\sT}{\mathbb{T}}

\newcommand{\proj}{{\mathcal{P}}}

\newcommand{\clo}[1]{\mathscr{#1}}
\newcommand{\EA}{\hrushovski{\sA}}
\newcommand{\ignore}[1]{}
\newcommand{\DEA}{\Decode{\hrushovski{\sA}}}
\DeclareMathOperator{\ari}{ar}
\newcommand{\arity}[1]{\ari(#1)}
\newcommand{\ES}{\hrushovski{\sS}}
\newcommand{\DX}{\Decode{{\sX}}}

\usepackage{xy}
\xyoption{all}

\newcommand{\set}[1]{\{#1\}}
\newcommand{\setm}[2]{\set{#1\mid#2}}
\newcommand{\card}[1]{|{#1}|}

\newcommand{\wu}{u}

\newcommand{\ww}{w}

\newcommand\PolComplexity{\textsc{P}}
\newcommand{\floor}[1]{\lfloor{#1}\rfloor}

\newcommand{\bigO}{\mathcal{O}}

\usepackage{tabularx}
\DeclareMathOperator{\Dec}{\mathbf{D}}
\DeclareMathOperator{\Blowen}{\overrightarrow{\mathbf{B}}}
\DeclareMathOperator{\Reducten}{\overrightarrow{\mathbf{R}}}
\DeclareMathOperator{\Hru}{\mathbf{E}}
\DeclareMathOperator{\Ext}{\overleftarrow{\mathbf{B}}}
\DeclareMathOperator{\Res}{\overleftarrow{\mathbf{R}}}
\DeclareMathOperator{\Com}{\mathbf{C}}

\newcommand{\Expansion}[1]{\Ext{#1}}
\newcommand{\Restriction}[1]{\Res{#1}}
\newcommand{\Completion}[1]{\Com{#1}}
\newcommand{\Decode}[1]{\Dec{#1}}

\newcommand{\coding}[2]{\Blowen_{#2}#1}
\newcommand{\reducten}[1]{\Reducten#1}
\newcommand{\hrushovski}[1]{\Hru{#1}}

\newcommand{\newage}[1]{{{#1}^+}} 

\newcommand{\sU}{\mathbb{U}}

\newcommand{\EC}{\hrushovski{\sC}}

\newcommand\Dwnu{Dissected weak near-unanimity}

\newcommand{\Projs}{\proj}
\renewcommand\proj{\clo{P}}

\begin{document}

\title[When symmetries are not enough]{When symmetries are not enough: a hierarchy of hard Constraint Satisfaction Problems}
\begin{abstract}

We produce a class of $\omega$-categorical structures with finite signature by applying a model-theoretic construction -- a refinement of the Hrushosvki-encoding -- to $\omega$-categorical structures in a possibly infinite signature. We show that the encoded structures retain  desirable algebraic properties 
of the original structures, but that the constraint satisfaction problems (CSPs) associated with these structures can be badly behaved in terms of computational complexity. This method allows us to systematically generate $\omega$-categorical templates whose CSPs  are  complete for a variety of complexity classes of arbitrarily high complexity, and $\omega$-categorical templates that show that membership in any given complexity class containing AC$^0$ cannot be expressed by a set of identities on the polymorphisms. It moreover enables us to prove that recent results about the relevance of topology on polymorphism clones of $\omega$-categorical structures also apply for CSP templates, i.e., structures in a finite language. Finally, we obtain a concrete algebraic criterion which could constitute a description of the delineation between tractability and NP-hardness in the dichotomy conjecture for first-order reducts of finitely bounded homogeneous structures.
\end{abstract}

\author[P. Gillibert]{Pierre Gillibert}
\author[J. Jonu\v{s}as]{Julius Jonu\v{s}as} 
\author[M. Kompatscher]{Michael Kompatscher}
\author[A. Mottet]{Antoine Mottet} 
\author[M. Pinsker]{Michael Pinsker}
  \thanks{%
Pierre Gillibert and Michael Pinsker have received funding from the Austrian Science Fund (FWF) through projects No {P27600} and  {P32337}. Michael Pinsker has also received funding from the Czech Science Foundation (grant No 18-20123S). 
Julius Jonu\v{s}as received funding from the Austrian Science Fund (FWF) through Lise Meitner grant No {M 2555}.
    Antoine Mottet has received funding from the European Research Council (ERC) under the European Union's Horizon 2020 research and innovation programme (Grant Agreement No 771005, CoCoSym).
Michael Kompatscher was supported by the grants PRIMUS/SCI/12 and UNCE/SCI/022 of Charles University Research Centre programs, as well as grant No 18-20123S of the Czech Science Foundation.
    }
    \thanks{A conference version of this material has appeared in the Proceedings of the 47th International Colloquium on Automata, Languages and Programming (ICALP 2020)~\cite{GJKMP-conf}.}

\maketitle


\section{Introduction}
\subsection{Constraint Satisfaction Problems}

The \emph{Constraint Satisfaction Problem}, or CSP for short, over a relational
structure $\sA$ is the computational problem of deciding whether a given finite
relational structure $\sB$ in the signature of $\sA$ can be homomorphically mapped into
$\sA$. The structure $\sA$ is known as the \emph{template} or
\emph{constraint language} of the CSP, and the CSP of the particular structure $\sA$ is
denoted by $\csp(\sA)$. A host of interesting computational problems can be modelled using
CSPs by choosing an appropriate 
template. For example, if $\sA$ is the structure with domain $\{0, 1\}$ and all
binary relations on the set $\{0, 1\}$, then $\csp(\sA)$ is precisely the
\textsc{2-SAT} problem, and if $\sA$ is the complete loopless graph on three 
vertices, then $\csp(\sA)$ is the $3$-colouring problem of graphs. Note that the template $\sA$ which defines the problem can also be infinite -- only the input structure $\sB$ is required to be finite in order to obtain a computational problem. Many well-known computational problems can be modelled, and can in fact only be modelled, using an infinite template. One example is the CSP of the order of the rational numbers $(\mathbb Q;<)$, which is equivalent to the problem of deciding whether a given finite directed graph is acyclic. The size of the signature of the template $\sA$, or in other words the number of its relations, is however generally required to be finite:  otherwise,  the encoding of its relational symbols might influence the computational complexity of $\csp(\sA)$, so that this complexity is not well-defined as per the structure $\sA$ itself. To emphasize the importance of this requirement, we shall henceforth call relational structures in a finite signature \emph{finite language structures} or, in statements about CSPs, \emph{CSP templates}.

\subsubsection{Finite-domain CSPs} The general aim in the study of CSPs is to understand the structural reasons for the hardness or the tractability of such problems. Structural reasons for tractability often take the form of some kind of symmetry; the goal then becomes to identify appropriate ways to measure the degree of symmetry of a problem, which should determine its complexity. This has been successfully achieved for CSPs of structures over a finite domain. As it turns out, every finite template either has, in a certain precise sense, as little symmetry as the $3$-colouring problem above, in which case its CSP is NP-complete; or it has more symmetry and its CSP is polynomial-time solvable, just like the $2$-SAT problem. 
This dichotomy result was conjectured by Feder and Vardi~\cite{FederVardiSTOC, FederVardi}, and proved, almost 25~years later, independently by Bulatov~\cite{BulatovFVConjecture} and
Zhuk~\cite{ZhukFVConjecture}.
\begin{theorem}[Bulatov~\cite{BulatovFVConjecture}, Zhuk~\cite{ZhukFVConjecture}]
  \label{theorem:BulZhu}
  Let $\sA$ be a finite CSP template. Then one of the following holds.
  \begin{itemize}
	\item $\sA$ is preserved by a 6-ary function $s$ on its domain satisfying the equation
	    \[
      s(x, y, x, z, y, z) = s(y, x, z, x, z, y)
    \] 
for all possible values $x,y,z$, and $\csp(\sA)$ is in P;
    \item $\sA$ is not preserved by such a function, and $\csp(\sA)$ is NP-complete.
  \end{itemize}
\end{theorem}
In this formulation of the dichotomy theorem, symmetry of $\sA$ is thus measured by the presence or absence of a $6$-ary function satisfying the above equation among the functions which preserve $\sA$. A finitary function on the domain of $\sA$ \emph{preserves $\sA$} if it is a homomorphism from the appropriate power of $\sA$ into $\sA$. Such functions are called \emph{polymorphisms} of $\sA$, and the set of all polymorphisms of $\sA$ is denoted by $\Pol(\sA)$. Polymorphisms are commonly perceived as `higher-order symmetries' of a relational structure akin to automorphisms; universally quantified equations which are satisfied by some polymorphisms of a structure are called \emph{identities} of the structure.

The fact that the identities of a structure $\sA$ are, for finite $\sA$, an
appropriate notion of measuring the degree of symmetry of $\sA$ in the context
of CSPs was already known long before the proof of the dichotomy
theorem~\cite{JBK}, and even before the equation of Theorem~\ref{theorem:BulZhu} was
discovered in~\cite{Siggers}; this fact is commonly referred to as the \emph{algebraic
approach to CSPs}. In fact, an equivalent formulation of
Theorem~\ref{theorem:BulZhu} can be given without mentioning this particular
equation: for a finite CSP template $\sA$, we have that $\csp(\sA)$ is in P if $\sA$ satisfies  some
non-trivial set of \emph{height~1 identities} (short: \emph{h1~identities}), and NP-complete otherwise.
An h1~identity is an identity of the form
$$
f(x_1,\ldots,x_m)=g(y_1,\ldots,y_n)\; ,
$$
where $f,g$ are function symbols and $x_1,\ldots,x_m,y_1,\ldots,y_n$ are variables; a set of such identities is non-trivial if it is not satisfied by all structures. The prefix `height~1' refers to the fact that $f,g$ are function symbols, rather than possibly nested terms of such symbols, as would be allowed in arbitrary identities. The insight that the complexity of the CSP of a finite structure only depends, up to polynomial-time reductions, on its h1~identities was obtained rather recently in~\cite{wonderland}.

\subsubsection{Infinite-domain CSPs}  
One advantage of modelling computational problems as CSPs is that certain
subclasses of CSPs are susceptible to a uniform mathematical approach, and the
two dichotomy proofs for templates over finite domains bear witness to its
power. The algebraic approach behind these proofs, however, does not require the
template to be finite; certain ``smallness" assumptions, to be discussed later,
are sufficient to allow for a natural adaptation. And although every
computational decision problem is polynomial time equivalent to the CSP of some
infinite template~\cite{BodirskyGrohe}, for a large and natural class of
infinite-domain CSPs, which considerably expands the class of finite-domain CSPs,
a similar dichotomy conjecture as for finite-domain CSPs has been formulated: namely, for 
the class of all \emph{first-order reducts of finitely bounded homogeneous
structures}. The following formulation of the conjecture is a slight
reformulation of the one  proposed in~\cite{wonderland}, and has been proved to
be equivalent to earlier and substantially different formulations
in~\cite{Barto:2017aa}. 

\begin{conjecture}
  \label{conjecture:infinitecsp}
  Let $\sA$ be a CSP template which is a first-order reduct of a finitely bounded homogeneous structure. Then
  one of the following holds.
  \begin{itemize}
    \item $\sA$ satisfies some non-trivial set of h1 identities locally, i.e., on every finite subset of its domain, and\/ $\csp(\sA)$ is in P; 
    \item there exists a finite subset of its domain on which $\sA$ satisfies no non-trivial set of h1 identities, and\/ $\csp(\sA)$ is NP-complete.
  \end{itemize}
\end{conjecture}
The conjectured P/NP-complete dichotomy has been demonstrated for numerous subclasses: for example for all CSPs in the class
MMSNP~\cite{MMSNP}, as well as for the CSPs of the first-order reducts of $(\mathbb{Q};
<)$~\cite{tcsps-journal}, of
any countable homogeneous graph~\cite{BMPP16} (including the random graph~\cite{BodPin-Schaefer-both}), and of the random poset~\cite{CSP-poset}. 

It is thus the \emph{local h1~identities}, i.e., the h1~equations which are
true for the polymorphisms of $\sA$ on finite subsets of its domain, which are
believed to be the right measure of symmetry of $\sA$ -- according to the
conjecture, they determine tractability or hardness of its CSP. We should
mention that similarly to finite templates, for all templates in the range of the
conjecture the corresponding CSP is in NP; moreover, NP-completeness of such CSPs follows
from the existence of a finite subset of the domain on which $\sA$ satisfies no
non-trivial set of h1~identities (the condition of the second item), by more
general results from~\cite{wonderland}. The missing part is hence a proof that
local symmetries in the form of non-trivial h1~identities on all finite sets
imply tractability of the CSP. This calls for the quest for the structural
consequences of this situation, and in particular, whether such local
symmetries imply global symmetries, i.e., identities which hold globally.

\subsection{What is symmetry?} 
\subsubsection{Local vs.~global symmetries, and topology} 


One of the differences between Theorem~\ref{theorem:BulZhu} and
Conjecture~\ref{conjecture:infinitecsp} is the consideration of local identities in the latter. An appropriate topology on the polymorphisms of a structure allows to reformulate the difference between local and global identities, as follows.

By definition, a set of h1~identities is non-trivial if it is not satisfied by all structures; this is equivalent to not being satisfied by the projections on a 2-element domain. Denoting the set of these projections by $\proj$, the set of all identities of $\sA$ is thus trivial if and only if there is a mapping $\xi\colon \Pol(\sA)\to\proj $ which preserves h1~identities; such mappings are called \emph{minion homomorphisms}. Similarly, $\sA$ satisfies only trivial h1~identities on some finite subset $F$ of its domain if and only if there is a mapping $\xi\colon \Pol(\sA)\to\proj$ which preserves all h1~identities which are true on $F$ (i.e., for values of the variables ranging within $F$ only); such mappings are called \emph{uniformly continuous minion homomorphisms}, and are indeed uniformly continuous with respect to the natural uniformity which induces the 
\emph{pointwise convergence topology} on finitary functions. Hence, the question whether non-trivial local h1 identities imply  non-trivial global h1 identities in a relational structure $\sA$ raises the following questions:
\begin{enumerate}[label = \rm{ (\arabic*)}]
  \item Is every minion homomorphism 
  from $\Pol(\sA)$ to $\proj$ uniformly continuous?
  \item Does the existence of a minion
  homomorphism from $\Pol(\sA)$ to $\proj$ imply the existence of a 
   uniformly continuous minion homomorphism  from $\Pol(\sA)$ to $\proj$? 
\end{enumerate} 

Clearly, when $\sA$ is finite, the distinction between local and global is void, and hence the answer to both (1) and (2) is positive. For general infinite $\sA$, the questions have been answered negatively 
in~\cite{BPP-projective-homomorphisms} and~\cite{BP19}, respectively. One
of the main problems of the mathematical theory of infinite-domain CSPs is to investigate which
assumptions on an infinite structure are sufficient to force the answer to the
questions to be positive -- in particular, whether the assumptions of 
Conjecture~\ref{conjecture:infinitecsp} imply such positive answer, in which
case we could omit the consideration of local rather than global identities in
its formulation. Question~(2) is the one truly relevant for CSPs; the first question 
is relevant in that a positive answer to~(1) provides a particularly strong
proof of  a positive answer to~(2).

\subsubsection{A uniform notion of symmetry?} The second difference between the dichotomy theorem for finite templates 
and Conjecture~\ref{conjecture:infinitecsp} is that in the former, tractability
is characterised by a concrete h1~identity of $\sA$. This difference is
essential and, in fact, tightly linked to the two questions above. The
importance of a fixed set of h1~identities lies in the fact that it provides
one uniform reason for tractability, which is not only pleasing aesthetically, 
but also paves the way to a uniform algorithm witnessing said tractability. The
connection with questions~(1) and~(2) above is that if the same fixed set of
h1~identities is true locally in a structure, then it is true globally, under a
mild assumption on the structure which largely comprises the range of
Conjecture~\ref{conjecture:infinitecsp} (that of $\omega$-categoricity -- see
Figure~\ref{fig:hierarchy}).

\begin{enumerate}[resume]
   \item Is there a fixed set $\Sigma$ of h1~equations such that every structure $\sA$ satisfying some 
         non-trivial h1 identities locally must satisfy $\Sigma$ globally? Failing that, is there a
         fixed ``nice" family $(\Sigma_n)_{n\geq 1}$ of sets of h1~equations such that every such structure must, on every finite subset of its domain, satisfy one of the sets of the family?
\end{enumerate} 
The answer of the first and stronger formulation of~(3) is positive in the finite case~\cite{Siggers}; in the general infinite case it is negative (folklore). 

\subsection{A hierarchy of smallness assumptions}
\begin{figure}
                \includegraphics[height=60mm]{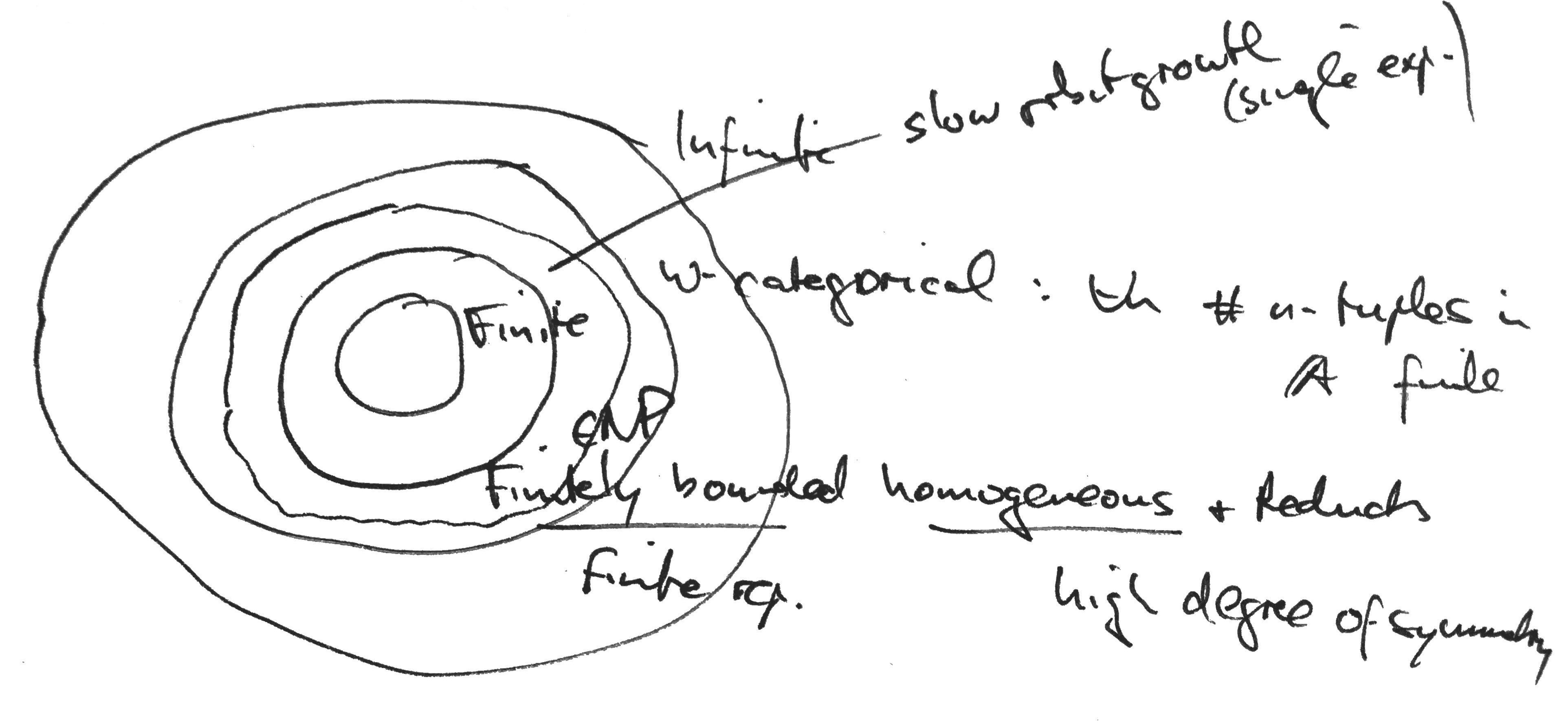}\caption{A vague representation of the hierarchy}
                \label{fig:hierarchy}
\end{figure}
A variety of restrictions of the class of all CSPs have been considered in the past in the search for a class for which a full complexity classification, and an understanding of the kind of symmetry which implies tractability, is feasible. Often, these restrictions take the form of ``smallness assumptions" on the relational structure defining the CSP. Such assumptions include restrictions on the size of the domain (Boolean, three elements, finite), or the range of Conjecture~\ref{conjecture:infinitecsp}. 

\subsubsection{The number of orbits of $n$-tuples.}\label{subsub:orbits}
The algebraic-topological approach outlined above, however, works in theory for a much larger class, namely the class of \emph{$\omega$-categorical} structures. A countable structure $\sA$ is $\omega$-categorical if its automorphism group $\Aut(\sA)$ acts with finitely many orbits on $n$-tuples, for all $n\geq 1$.
It is known that for $\omega$-categorical structures, the complexity of the CSP only depends on the polymorphisms of the template, viewed as a  topological clone~\cite{Topo-Birk}. Moreover, $\omega$-categoricity is sufficient to imply NP-hardness of the CSP if a structure satisfies no non-trivial h1~identities locally, i.e., on some finite set~\cite{wonderland}.

Although the class of $\omega$-categorical structures is far too vast to allow for a full complexity classification, two purely mathematical results nourished hope that this assumption, or strengthening thereof which are much milder than the assumptions of Conjecture~\ref{conjecture:infinitecsp},  could force the answers to Questions~(2) and~(3) to be positive.
The first result stated that under $\omega$-categoricity, local satisfaction of non-trivial h1~identities implies the global satisfaction of the (non-h1) \emph{pseudo-Siggers identity}~\cite{BartoPinskerDichotomy, BP19} -- a result entitled \emph{Topology is irrelevant}, in allusion to the local to global implication. The second result showed that if in the requirement for $\omega$-categoricity, the number of orbits of $\Aut(\sA)$ on $n$-tuples grows less than double exponentially in $n$ (a condition satisfied by all structures within the range of Conjecture~\ref{conjecture:infinitecsp}, and referred to as \emph{slow orbit growth} in this context), and $\sA$ is a \emph{model-complete core}, then the converse holds as well~\cite{Eq-oligo-CSP, Barto:2017aa}. The assumption of $\sA$ being a \emph{model-complete core} is not restrictive in the sense that every CSP of an $\omega$-categorical structure is equal to the CSP of an $\omega$-categorical structure which is a model-complete core.

\subsubsection{Finite language and finite relatedness.}
In~\cite{Bodirsky:2019aa, BMOOPW}, however, a counterexample to  Question~(2) was provided which was an $\omega$-categorical model-complete core with slow orbit growth -- a result referred to as \emph{Topology is relevant}, for obvious reasons. The counterexample lies clearly outside the range of Conjecture~\ref{conjecture:infinitecsp}, though -- and most importantly, it is not a CSP template since it  does not have a finite language! This drawback is in a sense inherent in the construction, since the structure provided is obtained as a ``generic superposition" of an infinite number of unrelated structures. 

The condition on structures of having a finite language is tacitly present in the context of CSPs by definition, and none of the above-mentioned smallness conditions, which have existed in model theory independently of the study of CSPs for many years, requires it. It is, indeed, a smallness condition itself, which however has not yet been utilized in the abstract mathematical theory of infinite-domain CSPs. On the other hand, in its role as a smallness condition, it has a long history in classical universal algebra: there, a finite algebra is called \emph{finitely related} if its term functions are the polymorphisms of a finite language structure.  For example, one of the most recent and spectacular results about finitely related finite algebras states that any such algebra in a \emph{congruence modular
variety} has \emph{few subpowers}~\cite{barto-cm}; consequently, any such algebra in a \emph{congruence distributive
variety} satisfies a \emph{near unanimity
identity}~\cite{BartoFinitelyRelated}. The polymorphisms of a CSP template always form, by definition, a finitely related algebra; in the light of the numerous deep results on such algebras, it is thus very well conceivable that the additional condition of a finite language on $\omega$-categorical structures could allow for stronger conclusions regarding their identities.

\subsection{Results} We refine a model-theoretic trick due to Hrushovski~\cite{HodgesLong} to encode $\omega$-categorical structures with an infinite signature into $\omega$-categorical finite language structures while preserving essential properties of the original, obtaining the following results.

\subsubsection{CSPs with local, but no global h1 identities} We provide a negative answer to question~(2) for finite language  structures by encoding the original counterexample from~\cite{Bodirsky:2019aa, BMOOPW}.

\begin{theorem}\label{thm:main:localnoglobalh1}
  There is an $\omega$-categorical finite-language structure $\sU$ with slow
  orbit growth such that there exists a minion homomorphism from $\Pol(\sU)$ to
  $\proj$, but no uniformly continuous one.
\end{theorem}

We also encode a counterexample to question~(1) from~\cite{BPP-projective-homomorphisms}  for \emph{clone homomorphisms}, which are mappings preserving arbitrary (not only h1) identities, into a finite language. Clone homomorphisms  appear in the original (and equivalent~\cite{Barto:2017aa, Eq-oligo-CSP}) formulation of Conjecture~\ref{conjecture:infinitecsp}, as given in~\cite{BPP-projective-homomorphisms,BartoPinskerDichotomy, BP19}. 

\begin{theorem}\label{thm:main:localnoglobal}
There exists an $\omega$-categorical finite-language structure $\sU$ 
  with a clone homomorphism from $\Pol(\sU)$ to $\proj$ that is not uniformly continuous.
\end{theorem}

\subsubsection{Dissected weak near-unanimity identities}
The negative answer to Question~(2) in~\cite{Bodirsky:2019aa, BMOOPW} provided an $\omega$-categorical structure with slow orbit growth which satisfies non-trivial h1 identities locally, but not globally. The local satisfaction of  non-trivial h1 identities was, however, shown indirectly, by means of the equivalence with a pseudo-Siggers identity mentioned at the end of Section~\ref{subsub:orbits}; no concrete set of local h1~identities was given. 
Here, we find concrete local h1~identities which prevent the structure from~\cite{Bodirsky:2019aa, BMOOPW}, as well as the encoded finite language structure in Theorem~\ref{thm:main:localnoglobalh1}, from having a uniformly continuous minion homomorphism into $\proj$. We call these identities \emph{\dwnu}. In fact, we obtain relatively general conditions on the symmetry of a structure which force \dwnu\ identities to be satisfied locally.

\begin{theorem}\label{thm:embtoPham}
  Let $\sU$ be a homogeneous structure. Let $F$ be a finite subset of $U$, and let $k> 1$. Assume the following two conditions hold. 
  \begin{enumerate}[label = \textrm{ (\roman*)}]
    \item Only relations of arity smaller than $k$ hold for tuples of elements in $F$;

    \item There is an embedding from $\sU^2$ into $\sU$.
  \end{enumerate}
  Then $\sU$ satisfies $(n,k)$ \dwnu\ identities on $F$ for all $n > k$.
\end{theorem}
This suggests a potential approach to the second (and weaker) statement of~(3) above, the first statement having been proven false,  even within the range of Conjecture~\ref{conjecture:infinitecsp}, in~\cite{Bodirsky:2019aa, BMOOPW}. 
\begin{question}\label{quest:dwnu}
Let $\sU$ be an $\omega$-categorical structure with slow orbit growth which satisfies non-trivial h1~identities locally. Does $\sU$ satisfy \dwnu\ identities locally?
\end{question}
We remark that \dwnu\ identities can be viewed as a generalization of \emph{weak near
unanimity} identities. Moreover, it follows
from~\cite{wnuf} and~\cite{wonderland} that if $\sU$ is a finite relational
structure satisfying non-trivial h1 identities, then $\sU$ satisfies weak near unanimity identities, giving a positive
answer to Question~\ref{quest:dwnu} in the finite case. We also note that the satisfaction of \dwnu\ identities has been proven for a large number of structures within the range of Conjecture~\ref{conjecture:infinitecsp} in~\cite{Eq-oligo-CSP, Barto:2017aa}.

\subsubsection{$\omega$-categorical CSP monsters}

The complexity of $\csp(\sA)$ is, for every $\omega$-categorical CSP template $\sA$, determined by $\Pol(\sA)$ viewed as a topological clone: if there exists a topological clone isomorphism $\Pol(\sA)\to\Pol(\sB)$ and $\sA$ and $\sB$ are $\omega$-categorical, then $\csp(\sA)$ and $\csp(\sB)$ are equivalent under log-space reductions~\cite{Topo-Birk}.
Conjecture~\ref{conjecture:infinitecsp} even postulates that for every template $\sA$ within its scope,
 membership of $\csp(\sA)$ in \PolComplexity{} only depends on the local h1 identities of $\sA$. The latter is equivalent to the statement that polynomial-time tractability is characterised by the global satisfaction of the single identity $\alpha s(x,y,x,z,y,z)=\beta s(y,x,z,x,z,y)$~\cite{Eq-oligo-CSP,BP19}.

Using our encoding, we prove that global identities do not characterise membership in \PolComplexity{} -- or, in fact, in any other non-trivial class of languages containing FO -- for the class of homogeneous CSP templates.
\begin{theorem}\label{thm:pseudo}
Let $\mathcal C$ be any class of languages that contains \textsc{AC}$^0$ and that does not intersect every Turing degree.
Then there is no countable set $\Sigma$ of identities such that for all homogeneous CSP templates membership in $\mathcal C$ is equivalent to the satisfaction of $\Sigma$.
\end{theorem}

The proof of Theorem~\ref{thm:pseudo} relies on encoding arbitrary languages as CSPs of homogeneous templates. These templates are obtained by applying the Hrushovski-encoding to structures which have only empty relations, but a complicated infinite signature.  On the way, we obtain a new proof of a result by Bodirsky and Grohe~\cite{BodirskyGrohe}.

\begin{theorem}\label{thm:complexity-main}
Let $\mathcal C$ be a complexity class such that there exist $\coNP^\mathcal C$-complete problems. Then there exists a homogeneous CSP template that satisfies non-trivial h1~identities and whose CSP is $\coNP^\mathcal C$-complete. Moreover, if $\PolComplexity\neq\coNP$, then there exists a CSP template with these algebraic properties whose CSP has $\coNP$-intermediate complexity.
\end{theorem}

In particular, Theorem~\ref{thm:complexity-main} gives complete problems for classes such as $\Pi_n^{\PolComplexity}$ for every $n\geq 1$, \PSPACE, \EXPTIME, or even every fast-growing time complexity class $\mathbf{F_\alpha}$ where $\alpha\geq 2$ is an ordinal (such as the classes \textsc{Tower} or \textsc{Ackermann}, see~\cite{DBLP:journals/toct/Schmitz16}). 

\subsection{Outline} The paper is organised in the following way -- definitions and general
notation are provided in Section~\ref{section:prem}. Our variant of Hrushovski's encoding and
its properties are described in Section~\ref{section:finitelang}. The encoding
is used on the structure from~\cite{Bodirsky:2019aa, BMOOPW}  in Section~\ref{section:local} in order to show Theorem~\ref{thm:main:localnoglobalh1} using Theorem~\ref{thm:embtoPham}, which is also proven there.
In Section~\ref{section:complexity}, we study the complexity of CSPs of templates produced with the encoding, proving in particular Theorems~\ref{thm:pseudo} and~\ref{thm:complexity-main}.
Finally, in Section~\ref{section:discon} we apply the encoding to the structure from~\cite{BPP-projective-homomorphisms} to prove Theorem~\ref{thm:main:localnoglobal}.


\section{Preliminaries}\label{section:prem}

\subsection{Relational structures and CSPs}
A \emph{relational signature}, or \emph{language}, is a family $\sigma=(R_i)_{i\in I}$ of symbols, each of which has a finite positive number $\arity{R_i}$, its \emph{arity}, associated with it. We write $R\in\sigma$ to express that the symbol $R$ appears in the signature $\sigma$. 
A \emph{relational structure} with signature $\sigma$, or a \emph{$\sigma$-structure}, is a pair $\sA=(A;(R_i^\sA)_{i\in I})$, where $A$ is a set called the \emph{domain} of the structure, and $(R_i^\sA)_{i\in I}$ is a family of relations on this  domain of the arities associated with the signature, i.e., each $R_i^\sA$ is a subset of $A^{\arity{R_i}}$.
Throughout this article we denote relational structures by blackboard bold letters, such as $\sA$, and their domain by the same letter in the plain font, such as $A$. We will tacitly assume that all relational structures, as well as their signatures, are at most countably infinite.

If $\sA,\sB$ are relational structures in the same signature $\sigma$, then a \emph{homomorphism} from $\sB$ to $\sA$ is a function $f\colon B\to A$ with the property that for all $R\in\sigma$  and every $(x_1, \ldots, x_{\arity{R}}) \in R^{\sB}$ we have that $(f(x_1), \ldots, f(x_{\arity{R}})) \in R^{\sA}$. The map $f$ is an \emph{embedding} if
it is injective and $(x_1, \ldots, x_{\arity{R}}) \in R^{\sB}$ if and only if $(f(x_1),
\ldots, f(x_{\arity{R}})) \in R^{\sA}$ for all  $R\in \sigma$ and all $x_1,\ldots,x_{\arity{R}}\in B$. An \emph{isomorphism} is a surjective embedding.

If $\sA$ is a relational structure in a finite signature, called a \emph{finite language structure} or a \emph{CSP template}, then $\csp(\sA)$ is the set of all finite structures $\sB$ in the same signature with the property that there exists a homomorphism from $\sB$ into $\sA$. This set can be viewed as a computational problem where we are given a finite structure $\sB$ in that signature, and we have to decide whether $\sB\in\csp(\sA)$. We are interested in the complexity of this decision problem relative to the size of the structure $\sB$ as measured by the cardinality of its domain.

\subsection{The range of the infinite CSP conjecture, and smallness conditions}  
A relational structure $\sC$ is \emph{homogeneous} if every isomorphism between finite induced substructures extends to an automorphism of the entire structure $\sC$. In that case, $\sC$ is uniquely determined, up to isomorphism, by its \emph{age}, i.e., the class of its finite induced substructures up to isomorphism. $\sC$ is \emph{finitely bounded} if its signature is finite and its age is given by a finite set ${\mathcal F}$ of forbidden finite substructures, i.e., the age consists precisely of those finite structures in its signature which do not embed any member of ${\mathcal F}$. A \emph{first-order reduct} of a relational structure $\sC$ is a relational structure $\sA$ on the same domain all of whose relations are first-order definable without parameters in $\sC$. Every reduct $\sA$ of a  finitely bounded homogeneous structure is \emph{$\omega$-categorical}, i.e., it is up to isomorphism the unique countable model of its first-order theory. Equivalently, its automorphism group $\Aut(\sA)$ is \emph{oligomorphic}: it has finitely many orbits in its componentwise action on $A^n$, for all finite $n\geq 1$. In fact,  if $\sA$ is a first-order reduct of a finitely bounded homogeneous structure, then the number of orbits in the action of $\Aut(\sA)$ on $A^n$ grows exponentially in $n$; in general, we say that structures where this number grows less than double exponentially in $n$ have \emph{slow orbit growth}. The CSP of any first-order reduct of a finitely bounded homogeneous structure is contained in the complexity class NP.


\subsection{Function clones and polymorphisms}
Let $C$ be a set. Then the map $\pi^n_i \colon C^n \to C$ given by
$\pi^n_i(x_1,\ldots, x_n) = x_i$, where $n\geq 1$ and $i \in \{1, \ldots,
n\}$, is called the \emph{$i$-th $n$-ary projection on $C$}. If $n,m\geq 1$, and $f \colon C^n
\to C$ and $g_1, \ldots, g_n \colon C^m \to C$ are functions, then we define the composition $f \circ (g_1,
\ldots, g_n) \colon C^{m} \to C$ by
\[
  (x_1, \ldots, x_m) \mapsto f(g_1(x_1,\ldots, x_m), \ldots, g_n(x_1, \ldots,
  x_m)).
\] 
A \emph{function clone} $\mathscr{C}$ on a set $C$ is a set of functions of finite arities on $C$ 
 which contains all projections and which 
is closed under composition.
The set $C$ is called the \emph{domain} of $\mathscr{C}$. The set of all projections on $C$ forms a function clone; for $|C| = 2$ we refer to this clone as the \emph{clone of projections} and denote it by $\Projs$.

A \emph{polymorphism} of a relational structure $\sA$ is a homomorphism from some finite power $\sA^n$ of the structure into $\sA$. The set of all polymorphisms of $\sA$ forms a function clone on $A$, and is called the \emph{polymorphism clone} of $\sA$ and denoted by $\Pol(\sA)$.

\subsection{Identities}
An \emph{identity} is a formal expression 
\begin{equation*}\label{equation:identity}
  s(x_1, \ldots, x_n) = t(y_1, \ldots, y_m)
\end{equation*} 
where $s$ and $t$ are abstract terms of function symbols, and $x_1,
\ldots, x_n, y_1, \ldots, y_m$ are the variables that appear in these terms. The identity is of \emph{height 1} if the terms $s$ and $t$ contain precisely one function symbol; in other words no
nesting of function symbols is allowed, and no term may be just a variable. A \emph{pseudo-h1 identity} is one obtained from an h1~identity by composing the terms $s$ and $t$ with distinct unary function symbols from the outside. The pseudo-Siggers identity mentioned in the introduction is an example. A \emph{pseudo-h1 condition} is a set of identities obtained from a set of h1~identities by composing all terms in it with distinct unary function symbols from the outside (if the same term appears twice, then each appearance gets a different unary function symbol).

We say that a set of identities $\Sigma$ is \emph{satisfied} in a function clone $\clo{C}$ if the function symbols which appear in $\Sigma$ can be mapped to functions of appropriate arity in $\clo{C}$ in such a way that all identities of $\Sigma$ become true for all possible values of their variables in the domain $C$ of $\clo{C}$.
If $F \subseteq C$ is finite, then we say that $\Sigma$ is satisfied \emph{locally on F} if the above situation holds where only values within $F$ are considered for the variables. The identities of a relational structure are defined as the identities of its polymorphism clone, and similarly we shall speak of identities of a relational structure on a finite subset of its domain, with the obvious meaning.

A set of identities is called \emph{trivial} if it is satisfied in any function clone; this is the case if and only if it is satisfied in the projection clone $\proj$. Otherwise, the set is called \emph{non-trivial}. We say that a function clone satisfies non-trivial identities \emph{locally} if it satisfies a non-trivial set of identities on every finite subset of its domain.
We shall use similar terminology for relational structures, and for h1~identities.

\subsection{Clone homomorphisms}
Let $\mathscr{C}$ and $\mathscr{D}$ be two function clones. Then a map $\xi
\colon \mathscr{C} \to \mathscr{D}$ is called a \emph{clone homomorphism} if it
preserves arities, projections, and composition. Preservation of projections means that it sends the $i$-th $n$-ary projection in $\clo{C}$ to the $i$-th $n$-ary projection in $\clo{D}$ for all $1\leq i\leq n$; preservation of composition means that for all $n,m\geq 1$, all $n$-ary $f \in \mathscr{C}$, and all $m$-ary
$g_1, \ldots, g_n \in \mathscr{C}$
\[
  \xi(f \circ (g_1, \ldots, g_n)) = \xi(f) \circ (\xi(g_1), \ldots, \xi(g_n)). 
\]
This is the case if and only if the map $\xi$ \emph{preserves identities}, i.e., whenever some functions  in $\mathscr{C}$ witness the satisfaction of some identity in $\mathscr{C}$, their images under $\xi$ witness the satisfaction of the same identity in $\mathscr{D}$.

A map $\xi \colon \mathscr{C} \to \mathscr{D}$ is called a 
\emph{minion homomorphism} (sometimes also called height~1 or h1~clone homomorphism) 
if it preserves arities and composition with projections; the latter meaning that for all for all $n,m\geq 1$, all $n$-ary $f \in \mathscr{C}$, and any projections
$\pi^m_{i_1}, \ldots, \pi^m_{i_n}\in\mathscr{C}$
\[
  \xi(f \circ (\pi^m_{i_1}, \ldots, \pi^m_{i_n})) = \xi(f) \circ (\pi^m_{i_1},
  \ldots, \pi^m_{i_n}). 
\]
This is the case if and only if the map $\xi$ preserves h1~identities.

The existence of clone and minion homomorphisms between function clones characterize their relative degree of global symmetry. Namely, for function clones $\mathscr{C}$ and $\mathscr{D}$, there exists a clone homomorphism from $\mathscr{C}$ into $\mathscr{D}$ if and    only if every set of identities which holds in $\mathscr{C}$ also holds in
    $\mathscr{D}$; and there exists a minion homomorphism from $\mathscr{C}$ into $\mathscr{D}$ if and
    only if every set of height 1 identities which holds in $\mathscr{C}$ also
    holds in $\mathscr{D}$. In particular,  there exists a clone homomorphism from $\mathscr{C}$ to $\proj$ if and only if every set of identities satisfied in $\mathscr{C}$ is trivial;  and there exists a minion homomorphism from $\mathscr{C}$ to $\proj$ if and    only if every set of h1~identities satisfied in $\mathscr{C}$ is trivial.

\subsection{Topology} The set of all finitary operations on a fixed set $C$ is naturally equipped with the \emph{topology of pointwise convergence}, under which forming the composition of operations is a continuous operation. A basis of open sets of this topology is given by the sets of the form
\[
  \{f\colon C^n \to C \mid f(a^i_1, \ldots, a^i_n) =
  b^i \text{ for all } 1 \leq i \leq m\} 
\]
where $n, m \geq 1$ and $a^i_1, \ldots, a^i_n, b^i \in C$ for all $1
\leq i \leq m$. The resulting topological space is a uniform space, in the case of $C$ being countable even a Polish space. Bearing the subspace topology, function clones then form natural topological objects. 
If $\mathscr{C}, \mathscr{D}$ are function clones, an arity preserving map $\xi \colon \mathscr{C} \to \mathscr{D}$ is then \emph{uniformly continuous} if and only if for every $n \geq 1$ and every finite $A
\subseteq D^n$ there exists a finite $B \subseteq C^n$ such that $f\mid_B =
g\mid_B$ implies that $\xi(f)\mid_A = \xi(g)\mid_A$. If the domain of $\mathscr{D}$ is finite, then this is the case if and only if for every $n \geq 1$  there exists a finite $B \subseteq C^n$ such that $f\mid_B =
g\mid_B$ implies $\xi(f) = \xi(g)$. Finally, a minion homomorphism $\xi\colon\mathscr{C}\to\proj$ is uniformly continuous if there exists a finite $B \subseteq C$ such that  $f\mid_{B^n} =
g\mid_{B^n}$ implies $\xi(f) = \xi(g)$, for all $n\geq 1$ and all $n$-ary $f,g\in\clo{C}$.

The \emph{local satisfaction} of identities and h1~identities can be characterised via \emph{uniformly continuous} clone and minion homomorphisms, respectively~\cite{Topo-Birk, uniformbirkhoff, wonderland}: there exists a uniformly continuous clone homomorphism from $\mathscr{C}$ to $\proj$ if and only if there exists a finite set $F\subseteq C$ such that any set of identities satisfied in $\mathscr{C}$ on $F$ is trivial; and there exists a uniformly continuous minion homomorphism from $\mathscr{C}$ to $\proj$ if and only if there exists a finite set $F\subseteq C$ such that any set of h1~identities satisfied in $\mathscr{C}$ on $F$ is trivial.

\subsection{pp-formulas and interpretations} Our encoded finite language structure will \emph{pp-interpret} the original structure, in the following sense.

A formula is \emph{primitive positive}, in short \emph{pp}, if it contains only existential quantifiers, conjunctions, equalities, and relational symbols.  If $\sA$ is a relational structure, then a
relation is \emph{pp-definable} in $\sA$ if it can be defined by a pp-formula in
$\sA$. It is well-known and easy to see that a relation that is pp-definable in $\sA$ is preserved by
every operation in $\Pol(\sA)$. 
A \emph{pp-interpretation} is a first-order interpretation in the sense of model theory where all the involved formulas are primitive positive: 
a structure $\sA$ \emph{pp-interprets} $\sB$ if a structure isomorphic to $\sB$ can be constructed from $\sA$ by pp-defining a subset $S$ of some finite power $A^n$, then pp--defining an equivalence relation $\sim$ on $S$, and then pp-defining relations on the equivalence classes of $\sim$. The number $n$ is referred to as the \emph{dimension} of the interpretation.

\subsection{Homogeneity and amalgamation, reducts, and homomorphic boundedness} Let $\mathcal{C}$ be a class of structures in some fixed relational signature which is closed under isomorphisms. We define the following properties the class $\mathcal{C}$ might have.
\begin{description}
  \item[Hereditary property (HP)] if $\sA \in \mathcal{C}$ and if $\sB$ is a
  substructure of $\sA$, then $\sB \in \mathcal{C}$. 

  \item[Amalgamation property (AP)] if $\sA, \sB, \sC \in \mathcal{C}$ and if
   $f_1 \colon \sA \to \sB$ and $f_2 \colon \sA \to \sC$ are embeddings,
  then there exist $\sD \in \mathcal{C}$ and embeddings $g_1 \colon \sB \to \sD$
  and $g_2 \colon \sC \to \sD$ such that $g_1 \circ f_1 =  g_2 \circ f_2$. 

  \item[Strong amalgamation property (SAP)] $\mathcal{C}$ satisfies AP and
  in addition $g_1$ and $g_2$ can be chosen to have disjoint ranges, except for the common values enforced by above equation.	
\end{description}

Homogeneous structures can be constructed from their age as follows.

\begin{theorem}[Fra\"{i}ss\'{e}'s Theorem, see~\cite{HodgesLong}] \label{thm:fraisse}
  Let $\sigma$ be a relational signature and let $\mathcal{C}$ be a 
  class of finite $\sigma$-structures which is closed under isomorphisms and satisfies HP and AP. Then there exists a
  $\sigma$-structure $\sA$ such that $\sA$ is countable, homogeneous, and the
  age of $\sA$ equals $\mathcal{C}$. Furthermore $\sA$ is unique up to isomorphism.
\end{theorem}

The structure $\sA$ in the theorem above is referred to as the 
\emph{Fra\"{i}ss\'{e} limit} of $\mathcal{C}$, and the class $\mathcal{C}$ as a 
\emph{Fra\"{i}ss\'{e} class}.

For a relational structure $\sA$ in signature $\sigma=(R_i)_{i\in I}$, and $J\subseteq I$, we call the structure $(A;(R^\sA_i)_{i\in J})$ in signature $\rho:=(R_i)_{i\in J}$ 
the \emph{$\rho$-reduct} of $\sA$; conversely $\sA$ is called an \emph{expansion} of any of its reducts, and a \emph{first-order expansion} of a reduct if all of its relations have a first-order definition in the reduct. We say that a structure is \emph{homogenizable} if it has a homogeneous first-order expansion. All $\omega$-categorical structures are homogenizable. A homogenizable structure $\sA$ has \emph{no algebraicity} if the age of any, or equivalently some, homogeneous first-order expansion of $\sA$ has SAP. 


%



Let ${\mathcal F}$ be a set of $\sigma$-structures, where $\sigma$ is a signature. A $\sigma$-structure $\sA$ is \emph{homomorphically bounded} by ${\mathcal F}$  if its age is defined by forbidding the structures in ${\mathcal F}$ homomorphically, i.e., the age of $\sA$ consists precisely of those finite structures in its signature which do not contain a homomorphic image of any member of ${\mathcal F}$ as an induced substructure.


\newcommand\Words{\Sigma^{\geq 2}}

\section{The Hrushovski-encoding}
\label{section:finitelang}

We present the encoding of an arbitrary homogenizable structure with no algebraicity into a CSP template, which will be the basis of our results. The construction is originally due to Hrushovski~\cite[Section 7.4]{HodgesLong}; it was designed to capture properties of the first-order theory and consequently the automorphism group of the original structure. We  refine his construction in order to also compare the polymorphism clones of the original structure and its encoded counterpart, and to control  the complexity of the CSPs of the produced templates. Our encoding will have the following main properties:
\begin{itemize}
\item The original structure can be uniquely decoded from its encoding: in fact, it will have a pp-interpretation (of dimension~$1$, using a trivial equivalence relation) in its encoding. This implies that the CSP of a finite language structure is not harder than the CSP of its encoding.
\item The encoding preserves several algebraic and model-theoretic properties of importance. For example, the original structure is $\omega$-categorical if and only if its encoding is; it has slow orbit growth if and only if its encoding does; the encoding has, like the original structure, no algebraicity; pseudo-h1~identities of the original structure transfer, to a certain extent, to the encoding; 
if the original structure is homomorphically bounded, then so is its encoding; and the finite structures  which homomorphically map into the encoding (i.e., its $\csp$) are related to  the structures  which homomorphically map into the original structure.
\end{itemize}

\subsection{The encoding} 
Let $\Sigma$ be a finite alphabet, and let $\Words$ denote the set of all
finite words over $\Sigma$ of length at least two. We are going to encode structures with a signature of the form $\rho=(R_w)_{w\in W}$, where $W\subseteq\Words$ and where the arity of each symbol $R_w$ equals the length $|w|$ of the word $w$. For the rest of this section we fix $\Sigma$ and $\rho$. Our goal is to encode any homogenizable $\rho$-structure $\sA$ with no
algebraicity into a structure $\hrushovski{\sA}$ (where $\hrushovski{}$ stands
for \emph{E. Hrushovski}) in a finite signature $\theta$ which is disjoint
from $\rho$ and only depends on $\Sigma$.

Note that by renaming its signature, and possibly artificially inflating the arity of its relations (by adding dummy variables),  any arbitrary structure with countably many relations can be given a signature of the above form without changing, for example, its polymorphism clone. However, the  encoding will depend on these modifications, and their effect on the algebraic and combinatorial properties of the encoding is beyond the scope of this article. The original encoding~\cite[Section 7.4]{HodgesLong} roughly corresponds to the case where $|\Sigma|=1$, and our generalization allows us to avoid such modifications for the structures we wish to encode, making in particular our complexity-theoretic results possible.

\begin{definition}\label{definition:theta}
  Let $\theta$ denote the signature $\{P, \iota, \tau , S\}\cup \{H_s \mid s
  \in \Sigma\}$, where $P$, $\iota$, $\tau$ are unary relation symbols, $H_s$ is
  a binary relation symbol for each $s \in \Sigma$, and $S$ is a $4$-ary
  relation symbol. For every signature $\sigma$ disjoint from $\theta$, define
  $\sigma^+$ to be the union $\sigma \cup \theta$.
\end{definition}

The encoding of a $\rho$-structure $\sA$ will roughly be obtained as follows: first, one takes a homogeneous first-order expansion $\sB$ in some signature $\sigma$; from its age  $K$, one defines a class $\newage{K}$ of finite structures in signature $\sigma^+$; and the encoding is the $\theta$-reduct of the Fra\"{i}ss\'{e} limit of $\newage{K}$. In order to define the class $\newage{K}$, we need the following definitions.

\begin{definition} \label{definition:validencode}
  Let $\sigma$ be a signature disjoint from $\theta$, let $\sA$ be a
  $\sigma^+$-structure, and let $w  \in \Words$. A tuple
  $(a_1, \ldots, a_{|w|}, c_1, \ldots, c_{|w|})$ of elements of $\sA$ is a \emph{valid
  $w$-code} in $\sA$ if the following hold:
  \begin{enumerate}[label = \textrm{(\alph*)}]
    \item $a_1, \ldots, a_{|w|} \in P^{\sA}$.
    
    \item $H_{w_i}^{\sA}(c_i, c_j)$ for all $1\leq i, j\leq |w|$ such that $j \equiv i + 1 \pmod {|w|}$.

    \item $\iota^{\sA}(c_1)$ and $\tau^{\sA}(c_{|w|})$.
    
    \item $S^{\sA}(a_i, a_j, c_i, c_j)$ for all $1\leq i,j\leq |w|$ with $i \neq j$.
  \end{enumerate}
\end{definition}

\begin{definition} \label{definition:separated}
  Let $\sigma$ be a signature disjoint from $\theta$, and let $\sA$ be a
  $\sigma^+$-structure.  Then $\sA$ is called \emph{separated} if
  \begin{enumerate}[label = \textrm{(\roman*)}]
    \item $H_s^{\sA}$ only relates pairs within $A\setminus P^{\sA}$ for all $s \in \Sigma$;

    \item $\iota^{\sA}, \tau^{\sA}$ are contained in $A \setminus P^{\sA}$;

    \item If $(a,b,c,d)\in S^{\sA}$, then $c, d \in A\setminus P^{\sA}$
    and $c \neq d$.
  \end{enumerate}
\end{definition}

It follows from (iii) above that in a separated structure a valid $w$-code can
only exist if $|w| \geq 2$; this is the reason for the exclusion of
unary relation symbols from $\rho$.

%
%

\begin{definition}\label{definition:newage}
  Let $\sA$ be a $\rho$-structure and
  let $\sB$ be a homogeneous first-order expansion of $\sA$ with signature
  $\sigma$ and 
   age $K$. Define $\newage{K}$ to be the class of all finite
  $\sigma^+$-structures $\sC$ with the following properties:
  \begin{enumerate}[label = \textrm{(\arabic*)}]
    \item The $\sigma$-reduct of the restriction of $\sC$ to $P^{\sC}$ is an
    element of $K$.

    \item $\sC$ is separated and for every $R\in\sigma$ the relation
    $R^{\sC}$ only relates tuples which lie entirely within $P^{\sC}$.

    \item If $R_w \in \rho$ and $(a_1, \ldots, a_{|w|}, c_1,
    \ldots, c_{|w|})$ is a valid $w$-code in $\sC$, then $(a_1,\ldots,
    a_{|w|}) \in R_w^{\sC}$.
  \end{enumerate}
\end{definition}

It turns out that $\newage{K}$ is indeed a Fra\"{i}ss\'{e} class in the case
where $K$ has both the HP and the SAP, or in other
words when $K$ is the age of a homogeneous structure with no algebraicity. We remark
here that the (non-strong) AP (as described in~\cite{HodgesLong}) is not
sufficient. To see this, suppose that $\newage{K}$ has the AP; we prove that $K$ has the SAP.
  Let $\sA,
\sB, \sC \in K$ be such that there are embeddings from $\sA$ into  $\sB$ and $
\sC$; without loss of generality, $\sA$ is an induced substructure of both $\sB$ and $\sC$, and the embeddings are the identity function on $A$. We define $\sA', \sB', \sC'\in\newage{K}$ with domains $A\cup\{c_1,c_2\}, B\cup\{c_1,c_2\}, C\cup\{c_1,c_2\}$, respectively, where $c_1,c_2$ are two new fixed distinct elements. In each of the three structures,
set $P$ to be interpreted as the original sets $A, B, C$ respectively. Fix $a \in A$, and let $(a, b, c_1, c_2)\in S^{\sB'}$  for every $b \in B \setminus A$; moreover, let all
the remaining relations from $\theta$ be empty. Then $\sA', \sB', \sC'\in\newage{K}$ and $\sA'$ embeds into both $\sB'$ and $\sC'$.   Hence, by the
assumption, there is an amalgam $\sD' \in \newage{K}$. Let $f_1 \colon \sB' \to
\sD'$ and $f_2 \colon \sC' \to \sD'$ be the embeddings witnessing the
amalgamation. Then for all $b\in B\setminus A$ we have $S^{\sD'}(f_1(a), f_1(b), f_1(c_1), f_1(c_2))$; however
$S^{\sD'}(f_2(a), f_2(c), f_2(c_1), f_2(c_2))$ does not hold, for any $c\in C\setminus A$. Finally, $f_1$
and $f_2$ agree on $\{c_1, c_2, a\}$, implying $f_1(b) \neq f_2(c)$. Therefore
the $\sigma$-reduct of $\sD'$ restricted to $P^{\sD'}$ is a strong amalgam of $\sA,
\sB, \sC$, proving that $K$ has the SAP.

\begin{lemma} \label{lemma:SAPencoding}
  Let $\sA$ be a $\rho$-structure and let $\sB$ be a homogeneous first-order expansion of $\sA$
  with age $K$.  If $K$ has the HP and the SAP, then $\newage{K}$ has the HP and the SAP as
  well.
 \end{lemma}

\begin{proof}
  It is routine to show that the HP for $K$ implies the HP for $\newage{K}$.

  In order to verify the SAP for $\newage{K}$, let $\sA, \sB, \sC\in \newage{K}$, and
  let $e_1\colon \sA \to \sB$, $e_2\colon \sA \to \sC$ be embeddings. Without loss of generality, $\sA$ is an induced substructure of $\sB$ and $\sC$, and the embeddings are both the identity function on $A$.
  Let us
  denote by $\sA',\sB', \sC'$ the $\sigma$-reducts of $\sA, \sB, \sC$ restricted to the subsets defined by $P$ in each of the structures. By
  definition $\sA',\sB'$ and $\sC'$ are elements of $K$. Thus there exist $\sD'
  \in K$ and embeddings $f_1': \sB' \to \sD' ,f_2': \sC'  \to \sD'$ that
  witness the SAP over $\sA'$; by the SAP, without loss of generality the
  domain $D'$ of $\sD'$ is just the union of $B'$ and $C'$, and $f_1,f_2$ the
  identity functions. Let $D:=B\cup C$. We define a structure $\sD$ on $D$ by
  setting $R^{\sD} :=
  R^{\sD'}$ for all  $R\in \sigma$, and $T^{\sD} = T^{\sB} \cup
 T^{\sC}$ for all $T \in \theta$. It is then 
 straightforward to check that the identity function is a $\sigma^+$-embedding
 of  both $\sB$ and $\sC$ into $\sD$.

  It remains to prove that $\sD \in \newage{K}$. By construction the conditions
  (1) and (2) of Definition~\ref{definition:newage} are satisfied in $\sD$. In
  order to see that also~(3) holds, suppose that $R_w\in\rho$ for some $w \in
  \Words$ and that $(a_1, \ldots, a_{|w|}, c_1, \ldots, c_{|w|})$ is a valid
  $w$-code in $\sD$. We claim that the elements of this code either lie
  completely in $B$, or in $C$.  Suppose there are $x,y \in \{a_1, \ldots,
  a_{|w|}, c_1, \ldots, c_{|w|}\}$ such that $x \in B \setminus C$
  and $y \in C \setminus B$.  Then $x \neq y$, and so there are $x_1, x_2, x_3,
  x_4\in D$ such that $x, y \in \{x_1, x_2, x_3, x_4\}$ and $S^\sD(x_1, x_2,
  x_3, x_4)$. However, this contradicts the definition of $S^\sD$ as the
  union of $S^\sB$ and $S^\sC$. Hence, without loss of generality, $(a_1, \ldots,
  a_{|w|}, c_1, \ldots, c_{|w|})$ is contained in $B$, such that
  $S^{\sB}(a_i,a_j,c_i,c_j)$ holds for all $i \neq j$. By definition $(a_1,
  \ldots, a_{|w|}, c_1, \ldots, c_{|w|})$ is a valid $w$-code in $\sB$.
  This implies $(a_1, \ldots, a_{|w|}) \in R^\sB$ and so $(a_1, \ldots,
  a_{|w|}) \in R^\sD$. Hence,~(3) holds for $\sD$.
\end{proof}

By Lemma~\ref{lemma:SAPencoding}, if $\sA$ has no algebraicity, and  $\sB$ is
a homogeneous first-order expansion of $\sA$ with age $K$, then $K^+$ has a Fra\"{i}ss\'{e}
limit, allowing us to define our encoding as follows.


\begin{definition} \label{definition:reducts}
  Let $\sA$ be a $\rho$-structure with no algebraicity and let $\sB$ be a
  homogeneous first-order expansion of $\sA$ with age $K$.  We define $\coding{\sA}{\sB}$,
  the \emph{encoding blow up} of $\sA$, to be the Fra\"{i}ss\'{e} limit of
  $\newage{K}$. Moreover, we define $\reducten{\sC}$ to be the  $\theta$-reduct
  of any structure $\sC$ with signature containing $\theta$. The
  $\emph{Hrushovski-encoding}$ $\hrushovski{\sA}$ is defined by
  $\hrushovski{\sA} := \reducten{ \coding{\sA}{\sB}}$. 
\end{definition}

It might be of help to the reader if we note that the operators used in the
encoding of a structure, i.e., $\coding{}{\sB}$ and $\reducten{\sC}$, bear
arrows from left to right; the operators used in the decoding of a structure, to be
defined later, bear arrows in the opposite direction.  Even though the
structure $\coding{\sA}{\sB}$ depends on the
particular homogeneous expansion $\sB$, we will show in
Proposition~\ref{prop:encoding_independence} that the Hrushovski-encoding
$\hrushovski{\sA}$ does not.
More precisely, if $\sB_1$ and $\sB_2$ are two homogeneous expansions of $\sA$, then
$\reducten{\coding{\sA}{\sB_1}}$ and $\reducten{\coding{\sA}{\sB_2}}$ are
isomorphic, justifying the notation $\hrushovski{\sA}$ for either of the two.
An illustration of relations holding in $\hrushovski{\sA}$ can be seen in Figure~\ref{F:ImageCode}.

By definition, the structure $\hrushovski{\sA}$ has the finite signature $\theta$.
In Section~\ref{sect:relationship}, we will investigate further properties of $\hrushovski{\sA}$; before that, we give the definitions which will allow us to decode a structure. 

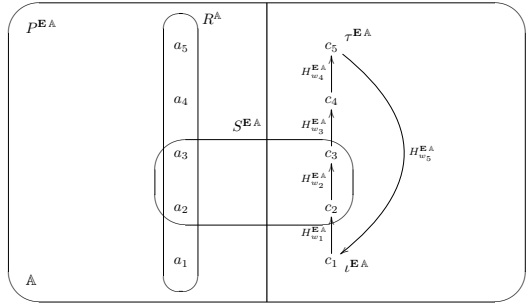
\begin{figure}[h]
\resizebox{7cm}{!}{
\setlength{\unitlength}{1cm}
\begin{picture}(12,7.2)
\put(6,3.5){\oval(12,7)}
\put(5.7,2.8){\oval(4.6,2.0)}
\put(4,3.5){\oval(0.8,6.5)}
\put(4.5,6.5){$R^{\sA}$}
\put(0.4,6.3){$P^{\hrushovski{\sA}}$}
\put(0.4,0.4){$\sA$}

\put(5.2,4){$S^{\hrushovski{\sA}}$}

\put(7.8,0.7){$\iota^{\hrushovski{\sA}}$}
\put(7.8,6.1){$\tau^{\hrushovski{\sA}}$}

\thicklines
\put(6,0){\line(0,1){7}}
\put(3.5,6){
$$
\xymatrix{
a_5 & & & c_5 \ar@/^4pc/[dddd]^{H_{w_5}^{\hrushovski{\sA}}} \\
a_4 & & & c_4 \ar[u]^{H_{w_4}^{\hrushovski{\sA}}} \\
a_3 & & & c_3 \ar[u]^{H_{w_3}^{\hrushovski{\sA}}}\\
a_2 & & & c_2 \ar[u]^{H_{w_2}^{\hrushovski{\sA}}}\\
a_1 & & & c_1 \ar[u]^{H_{w_1}^{\hrushovski{\sA}}}
}
$$
}
\end{picture}
}
\caption{The Hrushovski encoding $\EA$ of a structure $\sA$}
\label{F:ImageCode}
\end{figure}


\subsection{The decoding of an encoded structure}\label{sect:decode} Like the encoding of a structure, the decoding of a structure is a composition of two steps; first a \emph{decoding blow up}, and then a \emph{relativised reduct}.

\begin{definition} \label{definition:decoding}
  Let $\sC$ be a $\theta$-structure. Then the \emph{decoding blow up}
  $\Expansion{\sC}$ of $\sC$ is the expansion of $\sC$ in signature $\rho^+$,
  where for any symbol $R_w \in\rho$ the relation $R_w^{\Expansion{\sC}}$ is
  defined to consist of those tuples $(a_1,\dots,a_{|w|})$ for which
  there exist $c_1,\dots,c_{|w|} \in C$ such that
  $(a_1,\dots,a_{|w|},c_1,\dots,c_{|w|})$ is a valid $w$-code in $\sC$.

  For a structure $\sD$ in a signature containing $\rho^+$, the
  \emph{relativised reduct} $\Restriction{\sD}$ of $\sD$ is defined to be the
  $\rho$-reduct of $\sD$ restricted to $P^\sD$.   
  
  Finally, we set $\Decode{\sC} := \Restriction{\Expansion{\sC}}$, the
  \emph{decoding of $\sC$}, for any $\theta$-structure $\sC$.
\end{definition}

Table~\ref{table:operators} contains an informal summary of all operators, and
Figure~\ref{F:ResumeOperateurs} describes on which classes of structures they
operate. The operators bearing arrows are only auxiliary and will be useful in
the proofs; the operators we are truly interested in are $\hrushovski{}$ and
$\Decode{}$.  The last operator $\Completion{}$, assigning a finite
$\theta$-structure to a finite $\rho$-structure, will be used to compare the
finite structures which homomorphically map into $\sA$ with the CSP of its
encoding $\hrushovski{\sA}$.  It will be defined in
Section~\ref{section:homomorphismsandops}.

\begin{table}[h]
  \begin{tabularx}{\textwidth}{l|l|X}
    Operator & Name & Description   \\ \cline{1-3}

    $\coding{}{\sB}$ & encoding blow up & The first step in a Hrushovski-encoding, extends
    the domain and defines relations for the signature $\theta$ via a homogeneous expansion $\sB$ of the input. 
    \\ \cline{1-3}

    $\reducten{}$ & $\theta$-reduct & Returns the $\theta$-reduct of a structure. 
    \\ \cline{1-3}

    $\hrushovski{}$      & encoding &  Combines $\coding{}{\sB}$ and $\reducten{}$
    to obtain a $\theta$-structure from a $\rho$-structure. 
    \\ \cline{1-3}

    $\Expansion{}$   & decoding blow up & The first step in decoding a $\theta$-structure, it converts valid codes into corresponding 
    relations in $\rho$.
    \\ \cline{1-3}

    $\Restriction{}$ & relativised reduct & Restricts a structure to the set labelled by 
    $P$ and forgets the relations not in $\rho$.
    \\ \cline{1-3}

    $\Decode{}$      & decoding           &  Combines $\Restriction{}$ and
    $\Expansion{}$ to obtain the $\rho$-structure $\sA$ from the encoded $\theta$-structure
    $\hrushovski{\sA}$. 
    \\ \cline{1-3}

    $\Completion{}$  & canonical code    &     Defines in a canonical way a finite 
    $\theta$-structure from a finite $\rho$-structure in which every relation which holds in the input is witnessed by a valid
    code.
    \\ 
  \end{tabularx}
  \caption{The meaning of the operators}
  \label{table:operators}
  \end{table}

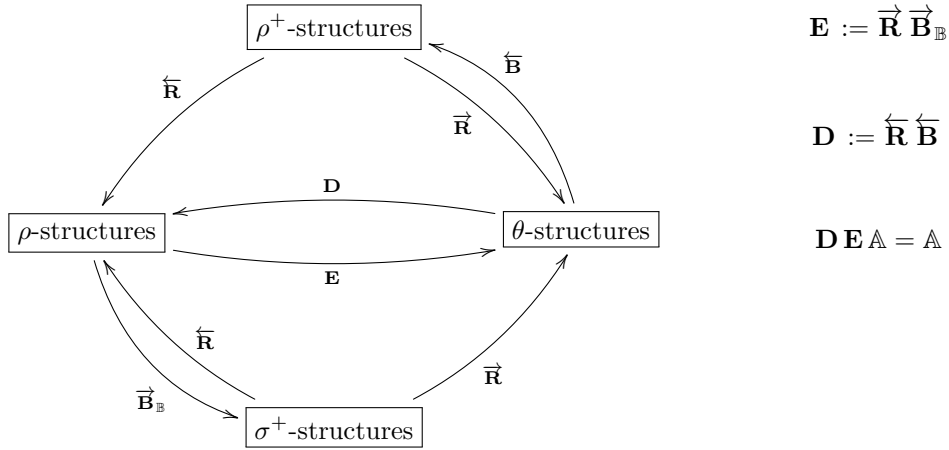
\begin{figure}[h]
$$
\xymatrix{
& \boxed{\rho^+\text{-structures}} \ar@/_1pc/[ddl]_{\Restriction{}}  \ar@/^1pc/[ddr]_{\reducten{}} & & \\ 
&  &  & \\
\boxed{\rho\text{-structures}}  \ar@/^1pc/@{<-}[rr]^{\Decode{}} \ar@/_2pc/[ddr]_{\coding{}{\sB}}    \ar@/_1pc/@{<-}[ddr]^{\Restriction{}} & & \ar@/^1pc/@{<-}[ll]^{\hrushovski{}}\boxed{\theta\text{-structures}}  
\ar@/_2pc/[uul]_{\Expansion{}} 
\\  &  &  & \\
& \boxed{\sigma^+\text{-structures}} \ar@/_1pc/[uur]_{\reducten{}} 
}\quad\quad
\xymatrix{
\hrushovski{} := \reducten{\coding{}{\sB}}\\
\Decode{} := \Restriction{\Expansion{}} \\
\Decode{\hrushovski{\sA}} = \sA
}
$$
\caption{Sources and destinations of operators.}
\label{F:ResumeOperateurs}
\end{figure}

\subsection{The relationship between $\sA$ and $\hrushovski{\sA}$}\label{sect:relationship}

We now investigate properties of the 
Hrushovski-encoding $\EA$ of a $\rho$-structure $\sA$, obtaining the following main results:
\begin{itemize}
  \item $\EA$ is independent of the first-order expansion of $\sA$ used on the
  way (Proposition~\ref{prop:encoding_independence});

  \item $\sA$ can be recovered from $\EA$ using the decoding:
  $\Decode{\EA}=\sA$ (Proposition~\ref{prop:expansionencoding}), and  in fact,
  the decoding is a pp-interpretation (Proposition~\ref{prop:ppint}); 
  
  \item $\EA$ is $\omega$-categorical if and only if $\sA$ is, and has slow
  orbit growth if and only if $\sA$ does (Proposition~\ref{prop:omegacat});

  \item There exists a uniformly continuous clone homomorphism $\xi$ from
  $\Pol(\EA)$ into $\Pol(\sA)$ (Proposition~\ref{prop:ucclonehomo}); moreover, if $\sA$ is $\omega$-categorical, then
  the injective functions in the image of $\xi$ are dense in the injective
  functions of $\Pol(\sA)$ (Corollary~\ref{cor:density});

  \item If $\sA$ is $\omega$-categorical, then the injective functions of $\Pol(\sA)$ essentially extend to functions in   $\Pol(\EA)$ (Lemma~\ref{lemma:encodingextensions}); consequently,  
   $\EA$ satisfies every 
  pseudo-h1 condition  which is satisfied in $\sA$ by injections
  (Proposition~\ref{prop:pseudo-h1-ids}).
\end{itemize}

In order to prove that $\hrushovski{\sA}$ is independent of the homogeneous first-order expansion used, we need the fact that
$\sA$ can be recovered from $\coding{\sA}{\sB}$ using $\Restriction{}$.

\begin{lemma}\label{lemma:A+hasA}
  Let $\sA$ be a $\rho$-structure with no algebraicity and let $\sB$ be a
  homogeneous first-order expansion of $\sA$ in signature $\sigma$.  Then the $\sigma$-reduct of the restriction of $\coding{\sA}{\sB}$
  to $P^{\coding{\sA}{\sB}}$ is isomorphic to $\sB$. Consequently, $\sA$ is
  isomorphic to $\Restriction{\coding{\sA}{\sB}}$.
\end{lemma}

\begin{proof}
  Let $\sigma$ be the signature of $\sB$. It follows from the definitions that
  the age of the $\sigma$-reduct of $\coding{\sA}{\sB}$ restricted to
  $P^{\coding{\sA}{\sB}}$ is contained in the age of $\sB$.  On the other hand,
  for every $\sC$ in the age of $\sB$, there is a structure $\sC'$, obtained by
  setting $P^{\sC'} = C$ and leaving the other relations empty, such that the
  $\sigma$-reduct of $\sC'$ restricted to $P^{\sC'}$ is $\sC$.  Hence the two
  ages are the same.  Since $\coding{\sA}{\sB}$ is homogeneous and since the
  only relations defined on the restriction of $\coding{\sA}{\sB}$ to $P$ are
  from $\sigma$, it follows that the $\sigma$-reduct of $\coding{\sA}{\sB}$
  restricted to $P^{\coding{\sA}{\sB}}$ is homogeneous.  Finally, by
  Theorem~\ref{thm:fraisse} it is isomorphic to $\sB$, and
  taking the $\rho$-reduct of the two structures yields the desired result.
\end{proof}

It follows from Lemma~\ref{lemma:A+hasA} that we may identify the structure
$\sA$ with $\Restriction{\coding{\sA}{\sB}}$.
\begin{center}
        \fbox{\bf From this point onward, we make this identification for the sake
        of simplicity. }
\end{center}
This means that we see $\coding{\sA}{\sB}$ as an expansion of $\sB$ by elements outside its domain (those not in the set named by $P$), and by relations in the signature $\theta$.

\begin{proposition}\label{prop:encoding_independence}
  Let $\sA$ be a $\rho$-structure with no algebraicity, and let $\sB_1$
  and $\sB_2$ be two homogeneous first-order expansions of $\sA$. Then 
  $\reducten{\coding{\sA}{\sB_1}}$ and $\reducten{\coding{\sA}{\sB_2}}$ are isomorphic. Consequently, $\EA$ is independent of the homogeneous first-order expansion used in its construction.
\end{proposition}
\begin{proof}
  First observe that if $\sB_1$ and $\sB_2$ are two homogeneous expansions of
  $\sA$ in signatures $\sigma_1$ and $\sigma_2$, respectively, then so is the
  structure in signature $\sigma_1\cup \sigma_2$ which has all the relations of
  both $\sB_1$ and $\sB_2$; hence, to prove the lemma it is sufficient to
  consider the case where $\sigma_1\subseteq \sigma_2$ and $\sB_1$ is the $\sigma_1$-reduct of $\sB_2$.

Since $\sB_1$ is an expansion of $\sA$, and since $\sB_2$ is first-order
definable in $\sA$, we have that $\sB_2$ is first-order definable in $\sB_1$.
By Lemma~\ref{lemma:A+hasA} we have that the $\sigma_1$-reduct of the
restriction of $\coding{\sA}{\sB_1}$ to the set named by $P$ is isomorphic to
$\sB_1$, and a similar statement holds for $\sB_2$. Let $\phi$ be a formula
over the language $\sigma_1$ which defines some relation of $\sB_2$ over
$\sB_1$, and denote by $\phi'$ the formula obtained from $\phi$ by restricting
all variables to $P$. We expand $\coding{\sA}{\sB_1}$ by all relations defined
via formulas of this form to obtain a structure $\sC$ in signature
$\sigma_2\cup\theta$. Being a first-order expansion of a homogeneous structure,
$\sC$ is homogeneous. By the above, the $\sigma_2$-reduct of the restriction of
$\sC$ to the set named by $P$ is isomorphic to $\sB_2$.

We claim that $\sC$ and $\coding{\sA}{\sB_2}$ have the same age. It is clear
that the age of $\sC$ is contained in the age of $\coding{\sA}{\sB_2}$: no
relations from $\rho$ have been added to $\coding{\sA}{\sB_1}$ in the
expansion, and hence the definition for being a member of the age of
$\coding{\sA}{\sB_2}$ is still satisfied by all finite substructures of $\sC$.
Conversely, let $\sF$ be a member of the age of $\coding{\sA}{\sB_2}$. Denote
by $\sF_2$ the restriction of $\sF$ to the set named by $P$. Then $\sF_2$
embeds into $\sC$; without loss of generality it is an induced  substructure
thereof. Denote by $\sF_1$ the $(\sigma_1\cup\theta)$-reduct of $\sF_2$. The
structure $\coding{\sA}{\sB_1}$ has a finite substructure $\sD$ whose
restriction to the set named by $P$ equals  $\sF_1$, and whose $\theta$-reduct
is isomorphic to the $\theta$-reduct of $\sF$ via an isomorphism which fixes
all elements of $\sF_1$. The structure induced in $\sC$ by the domain of $\sD$
then is isomorphic to $\sF$, proving the required inclusion.

Since $\sC$ and $\coding{\sA}{\sB_2}$ are homogeneous, they are isomorphic by
Theorem~\ref{thm:fraisse}. Hence, their $\theta$-reducts, which equal
$\reducten{\coding{\sA}{\sB_1}}$ and $\reducten{\coding{\sA}{\sB_2}}$
respectively, are also isomorphic. 
\end{proof}

Next, we  prove that $\Decode{}$ indeed decodes $\hrushovski{\sA}$. 

\begin{proposition}\label{prop:expansionencoding}
  Let $\sA$ be a homogenizable $\rho$-structure with no algebraicity.  Then $R^{\sA} =
  R^{\DEA}$ for all $R\in\rho$, and thus $\sA$ and
  $\Decode{\hrushovski{\sA}}$ are isomorphic.
\end{proposition}

\begin{proof}
  Let $\sB$ be a  
  homogeneous first-order expansion of $\sA$ in signature $\sigma$. Let $R_{w}$
  be any symbol of $\rho$.  First,
  note that $R_w^{\DEA} \subseteq
  R_w^{\sA}$ by Definition~\ref{definition:newage}~(3). 
  In order
  to prove the converse, let $(a_1,\ldots,a_{|w|})\in R_w^{\sA}$ be
  arbitrary, and let $\sF$ be the $\sigma$-structure induced by
  $\{a_1,\ldots,a_{|w|}\}$ in $\coding{\sA}{\sB}$.  We construct a $\sigma^+$
  structure $\sG$ by extending $\sF$ by distinct elements $c_1,\ldots,c_{|w|}$ and
  introducing relations from $\theta$ in such a way that
  $(a_1,\ldots,a_{|w|},c_1,\ldots,c_{|w|})$ is a valid $w$-code (but no other newly
  introduced tuples are related). It is routine to verify that $\sG \in K^+$.
  Since $\coding{\sA}{\sB}$ is homogeneous and $\sG$ is in the age of
  $\coding{\sA}{\sB}$, there exist $d_1,\ldots,d_{|w|}$ in $\coding{\sA}{\sB}$ such
  that the structure induced by $\{a_1,\ldots,a_{|w|},d_1,\ldots,d_{|w|}\}$ in
  $\coding{\sA}{\sB}$ is isomorphic to $\sG$. It follows that
  $(a_1,\ldots,a_{|w|})\in R_w^{\DEA}$. 
\end{proof}

\begin{proposition}\label{prop:ppint}
Let $\sC$ be a $\theta$-structure. Then $\Decode{\sC}$ has a pp-interpretation in $\sC$.
\end{proposition}
\begin{proof}
The dimension of the interpretation is~1, the pp-definable subset of $C$ is $P^\sC$, and the equivalence relation on $C$ can be chosen to be trivial. The definitions of the relations of $\Decode{\sC}$ are primitive positive.
\end{proof}

Next, we 
investigate the relationship of the orbits of $\Aut(\sA)$ with those of $\Aut(\EA)$, showing that $\omega$-categoricity and slow orbit growth are preserved by the encoding.

\begin{proposition}\label{prop:omegacat}
  Let $\sA$ be a homogenizable $\rho$-structure with no algebraicity.
  \begin{enumerate}
    \item $\sA$ is $\omega$-categorical if and only if $\EA$ is.

    \item Denote, for all $n\geq 1$, by
    $f(n)$ and $g(n)$ the (possibly infinite) number of orbits of $n$-tuples under the action of
    $\Aut(\sA)$ and $\Aut(\EA)$, respectively. Then $f(n)\leq g(n)$ for all
    $n\geq 1$, and $g(n) \leq 2^{6 |\Sigma| n^4} f(n)$. In particular, $\sA$ has slow
    orbit growth if and only if $\EA$ does.
\end{enumerate}
\end{proposition}

\begin{proof}
  Let us recall that a structure is $\omega$-categorical if for every $n \geq 1$ the number of $n$-ary orbits of its automorphism group is finite. Thus
  (2) implies (1).

  To prove (2), let $\sB$ be a
  homogeneous first-order expansion of $\sA$ in signature $\sigma$. Note that then $\Aut(\sA) = \Aut(\sB)$. Similarly, since $\coding{\sA}{\sB}$ and $\EA$ are first-order interdefinable, their automorphism groups are equal. Hence, it suffices to prove the statement for the homogeneous structures $\sB$ and $\coding{\sA}{\sB}$ instead of $\sA$ and $\EA$.

  Since $\coding{\sA}{\sB}$ is homogeneous, two tuples $(b_1,\ldots,b_n)$ and
  $(b_1',\ldots,b_n')$ lie in the same orbit of $\Aut(\coding{\sA}{\sB})$ if
  and only if the map that sends every $b_i$ to $b_i'$ is an isomorphism
  between the substructures of $\coding{\sA}{\sB}$ induced by
  $\{b_1,\ldots,b_n\}$ and $\{b_1',\ldots,b_n'\}$. In other words, the orbit of
  $(b_1,\ldots,b_n)$ under $\Aut(\coding{\sA}{\sB})$ is completely
  determined by its isomorphism type, i.e. the relations and equalities that hold for the entries of $(b_1,\ldots,b_n)$. The same
  statement is true for $\sB$. By Lemma~\ref{lemma:A+hasA} $\sB$ is equal to the restriction of $\coding{\sA}{\sB}$ to $P^{\coding{\sA}{\sB}}$, and every isomorphism type of $\sB$ corresponds to an isomorphism type of a tuple in $\coding{\sA}{\sB}$ that lies entirely in $P^{\coding{\sA}{\sB}}$. Therefore $g(n)\geq f(n)$.
  
  For the second inequality, we estimate the number of isomorphism types of $n$-ary tuples $(b_1,\ldots,b_n)$
  in $\coding{\sA}{\sB}$. First note, that there are $2^n$  partitions of the coordinates into elements that satisfy $P$ and elements
  that do not. Let us first count the number of isomorphism types for a fixed such
  partition with $m \geq 1$ many entries in $P$. Without loss of
  generality let it be the first $m$ entries and let $r := n-m$. By assumption,
  there are $f(m)$ many ways of introducing relations from $\sigma \cup \{=\}$
  on $(b_1,\ldots,b_m)$ so that it embeds into $\sB$, or equivalently, into $P^{\coding{\sA}{\sB}}$. There are less than
  $2^{r^2}$ ways of identifying the remaining $r$ entries. Counting
  further the different ways of introducing relations from $\theta$ on
  $(b_1,\ldots,b_n)$ such that the structure induced on $\{b_1,\ldots,b_n\}$ is
  separated, gives us an upper bound of $2^{r^2} \cdot 2^{2r + |\Sigma| r^2 +
  r^2m^2} f(m) \leq 2^{5 |\Sigma| n^4} f(m)$. By the monotonicity of $f$, this
  is smaller than $2^{5|\Sigma| n^4} f(n)$.  In the special case
 when $m = 0$ entries satisfy $P$, we analogously get an upper bound of $2^{5
  |\Sigma| n^4}$ orbits. Summing up over all  partitions, this gives us an upper bound $g(n) \leq 2^{n+5 |\Sigma| n^4}
  f(n) \leq 2^{6 |\Sigma| n^4} f(n)$, which concludes the proof.
\end{proof}

We now turn to the polymorphism clones of $\sA$ and $\EA$. An immediate consequence of Proposition~\ref{prop:ppint} is that if 
$\Pol(\hrushovski{\sA})$ satisfies non-trivial identities locally, then so does $\Pol(\sA)$. 

\begin{proposition} \label{prop:ucclonehomo}
  Let $\sA$ be a homogenizable $\rho$-structure with no algebraicity.  Then the map $\xi$
  that sends every $f \in \Pol(\hrushovski{\sA})$ to its restriction to $P^\sA$ 
  is a uniformly continuous clone homomorphism from $\Pol(\hrushovski{\sA})$ to
  $\Pol(\sA)$.
\end{proposition}

\begin{proof} 
  Any such restriction is a function on the domain $P^{\EA}$ of $\sA$. Since the relations of $\sA$ are pp-definable in $\EA$, they are preserved by the polymorphisms of $\EA$. Hence, the restriction to $P^{\EA}$ indeed defines a map from $\Pol(\hrushovski{\sA})$ to
  $\Pol(\sA)$. It clearly is a clone homomorphism and uniformly continuous.
\end{proof}

The next result demonstrates in particular that if $\sA$ is $\omega$-categorical, then for 
every injective $f \in \Pol(\sA)$ there exists a self-embedding $u$ of $\sA$ such
that $uf$ can be extended to a polymorphism of $\EA$.  We will, however,
require a more general and, \emph{h\'{e}las}, more technical statement than
this.

\begin{lemma}\label{lemma:encodingextensions} 
  Let $\sA$ be an $\omega$-categorical $\rho$-structure with no algebraicity,
  and let $\sB$ be a homogeneous first-order expansion of $\sA$ with signature
  $\sigma$. Furthermore, let $\sX$ be a separated $\theta$-structure. Then the following statements hold for all $k\geq 1$.
  \begin{enumerate}[label = \rm{(\arabic*)}]
  \item If $\sX$ is finite, then every injective homomorphism $f \colon (\DX)^k \to \sA$ extends to an embedding from $\sX^k$ to $\EA$.
  \item For every injective homomorphism $f \colon (\DX)^k \to \sA$ there exists an embedding $u \colon \coding{\sA}{\sB} \to \coding{\sA}{\sB}$ such that $u f$ extends to an embedding from $\sX^k$ to $\EA$.
  \item For every injective homomorphism $f \colon \sA^k\to\sA$ there exists an embedding $u \colon \coding{\sA}{\sB} \to \coding{\sA}{\sB}$ such that $u f$ extends to an embedding from $(\EA)^k$ into $\EA$.
\item $\sB^k$ embeds into $\sB$ if and only if
$(\coding{\sA}{\sB})^k$ embeds into $\coding{\sA}{\sB}$.
  \end{enumerate}
\end{lemma}

\begin{proof}

  (1) We start by defining a $\sigma^+$-structure $\sH$ and a map $h\colon X^k \to H$ as follows:
  The domain $H$ of $\sH$ is the disjoint union of the image of $f$ and of $X^k \setminus (P^{\sX})^k$, and $h$ is given by 
  \begin{align*}
    (x_1, \ldots, x_k) & \mapsto
    \begin{cases}
      f(x_1,\dots,x_k) &\text{if $(x_1,\dots,x_k)\in (P^{\sX})^k$,}\\
       (x_1,\dots,x_k) &\text{otherwise.}
    \end{cases}
  \end{align*}
  The $\sigma$-relations of $\sH$ are defined as the relations induced in the
  image of $f$ within $\sB$, and for every relation symbol $T \in \theta$ we
  set $T^\sH$ to be the image of $T^{\sX^k}$ under $h$. Since $f$ is injective,
  so is $h$. Moreover, since every $\theta$-relation of $\sH$ is defined as the
  image of the corresponding relation in $\sX^k$, $h$ is a $\theta$-embedding.

Next, we  show that $\sH$ lies in the age of $\coding{\sA}{\sB}$; it then
  follows directly from the homogeneity of $\coding{\sA}{\sB}$ that we can
  embed $\sH$ into $\coding{\sA}{\sB}$ fixing the image of $f$. Composing this
  embedding with the function $h$, we then obtain the desired expansion of $f$.
  In order to prove that $\sH$ lies in the age of $\coding{\sA}{\sB}$, it
  suffices to show that $\sH$ satisfies Definition~\ref{definition:newage}.
  Conditions (1) and (2) are routine to verify. In order to check (3), let
  $n\geq 1$ and let $(b_1, \ldots, b_n, \mathbf{c}_1,\ldots, \mathbf{c}_n)$ be
  a valid $w$-code in $\sH$ for some $w \in \Words$ such that $n = |w|$ and
  there is $R_w \in \rho$. We use the notation $\mathbf{c}_i = (c_{i, 1},
  \ldots, c_{i, k})$ for all $1\leq i\leq n$.  Since $f$ is injective, for all
  $1\leq i\leq n$ there is a unique $(a_{i, 1}, \ldots, a_{i,k}) \in X^k$ such
  that $b_i = f(a_{i, 1}, \ldots, a_{i, k})$. Since $h$ is an embedding with
  respect to $\theta$, we have that a tuple is a valid $w$-code in $\sH$ if and
  only if its pre-image under $h$ is a valid $w$-code in $\sX^k$. Thus
  $(a_{1, l}, \ldots, a_{n, l}, c_{1, l}, \ldots, c_{n, l})$ is a valid
  $w$-code in $\sX$ for all $1\leq l\leq k$. By the definition of the decoding operation $\Decode{}$, $R_w^{\DX}( a_{1,
  l}, \ldots, a_{n, l})$ for all $1\leq l\leq k$. Since $f$ preserves all
  relations of $\rho$ it follows that $R_w^{\sH}( b_{1}, \ldots,  b_{n})$. Thus
  $\sH$ satisfies Definition~\ref{definition:newage}.


  (2) Denote the domain of $f$ by $D^k$. By~(1), every finite substructure
  $\sF$ of $\sX^k$ can be mapped into $\EA$ by a homomorphism which extends the
  restriction of $f$ to $D^k\cap F^k$. By the $\omega$-categoricity of $\EA$, a
  standard compactness argument shows that the entire structure $\sX^k$ can be
  mapped into $\EA$ by a homomorphism $e$ whose restriction to $D^k$ is, in the
  language of~\cite{canonical}, \emph{locally equivalent} to $f$  with respect
  to $\Aut(\coding{\sA}{\sB})$: for every finite $S^k\subseteq D^k$ there
  exists an element $\alpha\in \coding{\sA}{\sB}$ such that $\alpha e$ and $f$
  agree on $S^k$. By Lemma 3 of~\cite{canonical} there exist two
  self-embeddings $v,u$ of $\coding{\sA}{\sB}$ such that $v \circ e = u \circ
  f$ on $D^k$. Setting $g = v \circ e$ then concludes the proof of~(2).
%

(3) This follows directly by setting $\sX:=\Hrushovski{\sA}$ in (2), since by Proposition~\ref{prop:expansionencoding} we have that $\sA$ is isomorphic to $\Decode{\Hrushovski{\sA}}$.

(4) Assume there is an embedding $f \colon \sB^k \to \sB$. Then $f$ is clearly also an injective homomorphism from $\sA^k$ to $\sA$, so by~(3) there exists an embedding $u \colon \coding{\sA}{\sB} \to \coding{\sA}{\sB}$ such that $uf$ extends to an embedding from  $(\EA)^k$ to $\EA$. This embedding is the desired embedding from $(\coding{\sA}{\sB})^k$ into $\coding{\sA}{\sB}$. For the opposite direction note that every restriction of any embedding of $(\coding{\sA}{\sB})^k$ into $\coding{\sA}{\sB}$ to $A^k$ is an embedding of $\sB^k$ into $\sB$.
\end{proof}


As a corollary to Lemma~\ref{lemma:encodingextensions}, we obtain that the
image of the uniformly continuous clone homomorphism from  $\Pol(\EA)$ to
$\Pol(\sA)$ which is given by restriction to $P^{\EA}$ is dense in the
injective part of $\Pol(\sA)$.
	
\begin{corollary}\label{cor:density}
  Let $\sA$ be an $\omega$-categorical $\rho$-structure without algebraicity. Then the set of restrictions of functions in $\Pol(\EA)$ is dense in the set of injective functions of $\Pol(\sA)$. 
\end{corollary}

\begin{proof}
Let $f \in \Pol(\sA)$ be injective, and denote its arity by $k$.
  Let $\sB$ be a homogeneous first-order expansion of
  $\sA$. By
  Lemma~\ref{lemma:encodingextensions}~(3), there is an embedding $u \colon \coding{\sA}{\sB} \to \coding{\sA}{\sB}$ such that $u\circ f$ can be
  extended to an embedding $g$ from $\EA^k$ into $\EA$.
  Let $F$ be a finite subset of $A$. Since $u$ is an
  embedding, and since $\coding{\sA}{\sB}$ is homogeneous, there exists $v \in
  \Aut(\coding{\sA}{\sB})$ such that $v$ agrees with $u$ on $f(F^k)$.
  Therefore, $f$ and $v^{-1} \circ g$ agree on $F^k$. Since $v^{-1} \circ g\in\Pol(\EA)$, the statement follows.
\end{proof}

As a further consequence of Lemma~\ref{lemma:encodingextensions}, the Hrushovski-encoding preserves the satisfaction of all pseudo-h1 conditions which are satisfied by injective functions.

\begin{proposition}\label{prop:pseudo-h1-ids}
  Let $\sA$ be an $\omega$-categorical $\rho$-structure with no algebraicity.
  Suppose that a pseudo-h1 condition $\Sigma$ is satisfied in
  $\sA$ by injections. Then $\Sigma$ is also satisfied in $\hrushovski{\sA}$.
\end{proposition}

\begin{proof}
  We fix a homogeneous first order expansion $\sB$  of $\sA$, and denote its signature by 
  $\sigma$. For any structure $\sX$, define an equivalence relation
  $\sim_{\sX}$ on $\Pol(\sX)$ by setting $f_1 \sim_{\sX} f_2$ if and only if $f_1$ and
  $f_2$ are of the same arity, and for every finite subset $F$ of $X$ there
  exist injective endomorphisms $e_1, e_2$ of $\sX$ such that $e_1 \circ f_1 = e_2 \circ f_2$ on
  $F^{\arity{f_1}}$. Recall that, by Lemma~\ref{lemma:encodingextensions}~(3), for every
  injective $f \in \Pol(\sA)$ there exist a self-embedding $u_f$ of $\coding{\sA}{\sB}$ and $\overline{f} \in \Pol(\hrushovski{\sA})$ such
  that $\overline{f}$ extends $u_f \circ f$. We fix such $\overline{f}$ and $u_f$; neither has to be unique.

  We begin by showing that for all injective $f_1,f_2\in\Pol(\sA)$ we have $f_1 \sim_\sA f_2$ if and only if $\overline{f_1}
  \sim_{\EA} \overline{f_2}$. If $\overline{f_1}
  \sim_{\EA} \overline{f_2}$, then $f_1 \sim_\sA f_2$ by Proposition~\ref{prop:ucclonehomo}. For the other direction, assume that $f_1 \sim_{\sA} f_2$, and denote the arity of $f_1$ and $f_2$ by $n$.
   Let $F$ be a finite subset of the domain of 
  $\coding{\sA}{\sB}$. By assumption, there are injective endomorphisms $e_1, e_2$ of $\sA$ such
  that $e_1 \circ f_1 = e_2 \circ f_2$ on $(F \cap A)^n$. Since
  $\coding{\sA}{\sB}$ is homogeneous and $u_{e_1}$, $u_{e_2}$ are self-embeddings of this structure,
  there are $v_1, v_2 \in \Aut(\coding{\sA}{\sB})$ such that both $v_1 \circ
  u_{f_1}$ and $v_2 \circ u_{f_2}$ act as the identity on the image of $(F \cap A)^n$ under $f_1$ and $f_2$, respectively. Consider the map $\psi \colon \overline{e_1} \circ
  v_1 \circ \overline{f_1}(x_1, \ldots, x_n) \mapsto \overline{e_2} \circ v_2
  \circ \overline{f_2}(x_1, \ldots, x_n)$ between the substructures of
  $\coding{\sA}{\sB}$ induced by the images of $F^n$ under $ \overline{e_1} \circ v_1 \circ
  \overline{f_1}$ and $\overline{e_2} \circ v_2 \circ \overline{f_2}$, respectively. By the injectivity of all involved functions, this map is well-defined. Note that all the functions appearing in the definition of $\psi$ are
embeddings with respect to the relations of $\theta$, and so $\psi$ is an isomorphism with respect to this signature. On the other
  hand, the relations of the signature $\sigma$ in $\coding{\sA}{\sB}$ are only non-empty on
  $P^{\coding{\sA}{\sB}} = A$. Moreover, $\overline{e_1} \circ v_1 \circ
  \overline{f_1}= u_{e_1} \circ e_1 \circ f_1$ and $\overline{e_2} \circ v_2
  \circ \overline{f_2}= u_{e_2} \circ e_2 \circ f_2$ on $(F \cap A)^n$. Since
  $e_1 \circ f_1 = e_2 \circ f_2$ on $(F \cap A)^n$ and since $u_{e_1}$ and $u_{e_2}$
  are embeddings with respect to $\sigma^+$, $\psi$ is even an isomorphism with respect to $\sigma^+$. Hence, there
  exists $w \in \Aut(\coding{\sA}{\sB})$ extending $\psi$, and thus $ (w \circ
  \overline{e_1} \circ v_1) \circ f_1 = (\overline{e_2} \circ v_2) \circ f_2$
  of $F^n$. Hence, 
  $\overline{f_1} \sim_{\EA} \overline{f_2}$.

It follows from the proof of~\cite[Lemma~3]{canonical} that in an $\omega$-categorical structure $\sX$, $f_1\sim_\sX f_2$ implies the existence of injective endomorphisms $e_1,e_2$ of $\sX$ such that $e_1\circ f_1=e_2\circ f_2$.

Suppose that $\Sigma$ is a  pseudo-h1 condition 
  satisfied in $\sA$ by injections. Let $v_1 \circ f_1(x_1, \ldots,
  x_{n}) = v_2 \circ f_2(y_1, \ldots, y_{m})$ be one of the identities from
  $\Sigma$, and, for the sake of brevity, identify the functions satisfying the
  identity with the symbols $f_1$, $f_2$, $v_1$, and $v_2$. Since the functions
  are injective, it follows that $\{x_1, \ldots, x_n\} = \{y_1, \ldots, y_m\}$.
  Let $z_1, \ldots, z_k$ be any enumeration of the variables in $\{x_1, \ldots,
  x_n\}$, and define $g_1(z_1, \ldots, z_k) := f_1(x_1, \ldots, x_n)$ and
  $g_2(z_1,\ldots, z_k) := f_2(y_1, \ldots, y_m)$. 
  Note that $\overline{g_1}$ can be chosen so that $\overline{g_1}(z_1,
  \ldots, z_k) = \overline{f_1}(x_1, \ldots, x_n)$, because $\overline{f_1}(x_1, \ldots, x_n)$ is an extension of $u_{f_1}\circ g_1$. A similar statement holds for $\overline{g_2}$. Since $g_1 \sim_{\sA} g_2$
  by virtue of the satisfied pseudo-h1 identity, we have $\overline{g_1} \sim_{\EA}
  \overline{g_2}$. Hence, there exist endomorphisms $w_1, w_2$ of $\EA$ such that $w_1 \circ \overline{g_1} = w_2 \circ
  \overline{g_2}$, and thus $w_1 \circ \overline{f_1}(x_1, \ldots, x_n) = w_2
  \circ \overline{f_2}(y_1, \ldots, y_m)$. Therefore, $\Sigma$ is satisfied in
  $\EA$.
\end{proof}

\subsection{Homomorphisms and the encoding}
\label{section:homomorphismsandops}
We now examine the relationship between the finite structures that homomorphically map into a structure $\sA$ with those that homomorphically map into its encoding $\EA$; the latter is precisely $\csp(\EA)$. This will be particularly
  relevant in Section~\ref{section:complexity} where we investigate the complexity of CSPs of structures encoded with the Hrushovski-encoding.

In the following definition, we assign to every $\rho$-structure $\sC$ a $\theta$-structure  $\Completion{\sC}$ in such a way that the original structure can be recovered. Contrary to the operator $\hrushovski{}$ this operator $\Completion{}$ is however mostly intended for finite structures; applied to a finite structure, it yields a finite structure.

\begin{definition} \label{definition:ccode}
  Let $\sC$ be a $\rho$-structure. Then the \emph{canonical code}
  $\Completion{\sC}$ of $\sC$ is the $\theta$-structure with underlying set 
  \[ 
    C \cup \{(w, \mathbf{t},i) \mid w \in \Words, \;R_w\in\rho,\;
    \mathbf{t}\in R_w^\sC, \text{ and } 1\leq i\leq |w| \}
  \] 
  and relations 
  \begin{itemize}
    \item $P^{\Completion{\sC}} = C$;

    \item $H_s^{\Completion{\sC}} = \{ \left((w, \mathbf{t},i),\;
	  (w, \mathbf{t},j)\right) \mid w \in \Words, \; R_w \in \rho,\;
	  \mathbf{t}\in R_w^\sC, \; w_i =s, \text{ and }i\equiv j+1\mod
	  |w|\}$ for all $s \in \Sigma$;

    \item $\iota^{\Completion{\sC}} = \{(w, \mathbf{t},1) \mid w \in \Words,
    \; R_w \in \rho,\; \mathbf{t} \in R_w^\sC \}$ and $\tau^{\Completion{\sC}} =
    \{(w, \mathbf{t}, |w| ) \mid w \in \Words, \; R_w \in \rho,\;
    \mathbf{t} \in R_w^\sC\}$;

    \item $S^{\Completion{\sC}} = \{ (t_i,t_j,(w, \mathbf{t},i),(w,
    \mathbf{t},j)) \mid w \in \Words, \; R_w \in \rho,\;
    \mathbf{t}=(t_1,\dots,t_{|w|})\in R_w^\sC, \text{ and } i \neq j \}$.
  \end{itemize}
\end{definition}
 
\begin{lemma}\label{lemma:remarks}
  The following statements hold.
  \begin{enumerate}[label = \rm{(\arabic*)}]
    \item $\sB=\Decode{\Completion{\sB}}$
    for every $\rho$-structure $\sB$;

    \item Let $\sB$ and $\sC$ be two $\theta$-structures. Then $f\colon
            \sB\to\sC$ is a homomorphism if and only if $f\colon\Expansion{\sB}
            \to \Expansion{\sC}$ (i.e., $f$ viewed as a function from $\Expansion{\sB}$ to $\Expansion{\sC}$)  is a homomorphism;
		
    \item  If $\sA$ is a $\rho$-structure, and $\sB$ a
            $\theta$-structure, then there exists a homomorphism from $\sA$ to $\Decode{\sB}$
            if and only if there exists a homomorphism from $\Completion{\sA}$ to $\sB$.
  \end{enumerate}
\end{lemma}

\begin{proof}
  (1) If $w \in \Sigma^{\geq 2}$, $R_w \in \rho$, and $\mathbf{t} = (t_1, \ldots,
  t_{|w|}) \in R_w^{\sB}$, then $(t_1,\dots,t_{|w|},(w, \mathbf{t},1),\dots,(w,
  \mathbf{t},{|w|}))$ is a valid $w$-code in $\Completion{\sB}$. The rest
  follows immediately from the definitions.

  (2) If $f\colon
  \Expansion{\sB} \to\Expansion{\sC}$ is a homomorphism, then clearly so is $f\colon \sB\to\sC$ since $\sB$ and $\sC$ are reducts of $  \Expansion{\sB}$ and $\Expansion{\sC}$, respectively. The other direction
  follows from the fact that the relations of $\Expansion{\sB}$ and $\Expansion{\sC}$ have primitive positive definitions in $\sB$ and $\sC$, respectively, and are thus preserved by homomorphisms.

  (3) Let $f\colon \Completion{\sA}\to\sB$ be a homomorphism. Then by~(2),
  $f\colon\Expansion{\Completion{\sA}} \to \Expansion{\sB}$ is a homomorphism
  as well, and so its restriction to $P^\sA$ is a homomorphism from
  $\Decode{\Completion{\sA}}=\sA$ to $\Decode{\sB}$. In order to show the other
  implication, let $f \colon \sA \to \Decode{\sB}$ be a homomorphism, let $w \in \Words$ such that $R_w
  \in \rho$, and let $\mathbf{t} = (t_1,\ldots, t_{|w|}) \in R_w^{\sA}$.
  Then $(f(t_1),\ldots, f(t_{|w|})) \in R_w^{\Decode{\sB}}$, and so there
  exist $c_{w, \mathbf{t}, 1}, \ldots, c_{w, \mathbf{t}, {|w|}} \in B$ such
  that $(f(t_1), \ldots, f(t_{|w|}), c_{w, \mathbf{t}, 1}, \ldots,
  c_{w, \mathbf{t}, {|w|}})$ is a valid $w$-code in $\sB$. Set $g \colon
  \Completion{\sA} \to \sB$ to be the extension of $f$ defined by
  $g((w, \mathbf{t}, i)) = c_{w, \mathbf{t}, i}$ for all $1\leq i\leq
  {|w|}$.  It is routine to verify that $g$ is a homomorphism. 
\end{proof}

We are now ready to compare the structures which homomorphically map into a structure with the $\csp$ of its encoding. 
\begin{proposition}\label{prop:otherremarks}
  Let $\sA$ be a $\rho$-structure with no algebraicity. Let $\sX$ be a separated $\theta$-structure, and let $\sY$ be a $\rho$-structure. 
  \begin{enumerate}
    \item If there exists a homomorphism from $\sX$ to $\hrushovski{\sA}$, then there exists a
            homomorphism from $\Decode{\sX}$ to $\sA$.

    \item If $\sA$ is $\omega$-categorical and there exists an injective homomorphism from $\Decode{\sX}$ to $\sA$, then
            there exists an injective homomorphism from $\sX$ to $\hrushovski{\sA}$.           
            
     \item $\sY$ has a homomorphism into $\sA$ if and only if $\Completion{\sY}$ has a homomorphism into $\EA$.
  \end{enumerate}
\end{proposition}

\begin{proof}
  (1) If $f\colon \sX\to\hrushovski{\sA}$ is a homomorphism, then $f\colon \Expansion{\sX}\to\Expansion{\hrushovski{\sA}}$ is a homomorphism, by Lemma~\ref{lemma:remarks}~(2). Its restriction to $P^\sX$ then is a homomorphism from $\Decode{\sX}$ to $\DEA$.  By Proposition~\ref{prop:expansionencoding}, $\sA=\DEA$.
  
  (2) 
  This follows directly from Lemma~\ref{lemma:encodingextensions} (2).
  
  (3) Assume first that $\Completion{\sY}$ has a homomorphism into $\EA$. Then  $\Decode{\Completion{\sY}}$ has a homomorphism into $\sA$, by~(1). By Lemma~\ref{lemma:remarks}~(1), $\Decode{\Completion{\sY}}=\sY$. For the other direction, assume that $\sY$ has a homomorphism into $\sA$. By Proposition~\ref{prop:expansionencoding}, $\sA=\DEA$, and so application of Lemma~\ref{lemma:remarks}~(3) shows that $\Completion{\sY}$ has a homomorphism into $\EA$.
\end{proof}


The properties from Proposition~\ref{prop:otherremarks} are enough to give a concrete description of $\csp(\Hrushovski{\sA})$ when $\sA$ is homomorphically bounded.

\begin{proposition}\label{prop:emb-forb}
  Let $\sA$ be a homogenizable $\rho$-structure with no algebraicity which is homomorphically bounded by a set $\mathcal{G}$ of $\rho$-structures.
  Let $\sX$ be a 
  $\theta$-structure. Then the following are equivalent.
  \begin{enumerate}[label = \rm{(\arabic*)}]
    \item There exists an embedding of $\sX$ into $\hrushovski{\sA}$;

    \item There exists a homomorphism from  $\sX$ to $\hrushovski{\sA}$;

    \item $\sX$ is separated and for all $\sG \in \mathcal{G}$ we have that there
            exists no homomorphism from $\Completion{\sG}$ to $\sX$. 
  \end{enumerate}
\end{proposition} 
\begin{proof}
  (1) $\implies$ (2) is trivial.

  (2) $\implies$ (3). Assume there exists a homomorphism $f\colon \sX \to
  \hrushovski{\sA}$. Since $\hrushovski{\sA}$ is separated, it follows that  
  $\sX$ is also separated, since this property is expressed only by negations of relations, which in turn are preserved under preimages of homomorphisms. Now, for the sake of contradiction, assume that there exists a
  homomorphism $g\colon \Completion{\sG} \to \sX$ for some $\sG \in \mathcal{G}$. Then 
composing $f$ with $g$, we obtain a homomorphism from $\Completion{\sG}$ to $\hrushovski{\sA}$. By
  Lemma~\ref{lemma:remarks}~(3), $\sG$ maps homomorphically into
  $\Decode{\hrushovski{\sA}}$, which is isomorphic to $\sA$ by
  Proposition~\ref{prop:expansionencoding}. Thus $\sG$ maps homomorphically to
  $\sA$, which is a contradiction.
 	
  (3) $\implies$ (1).  Let $\sB$ be a homogeneous first-order expansion of $\sA$ with
  signature $\sigma$. For every $\sG \in \mathcal{G}$, it follows from
  Lemma~\ref{lemma:remarks}~(3) that there exists a homomorphism from $\sG$ to $\Decode{\sX}$ if and only if there exists a homomorphism from $\Completion{\sG}$ to $\sX$; the latter, however, contradicts our assumption. Therefore $\Decode{\sX}$ embeds into $\sA$; for the sake of simplicity assume that $\Decode{\sX}$ is a substructure of $\sA$.
  Now, let $\sY$ be an expansion of $\sX$ to a $\sigma^+$-structure such that the  $\sigma$-reduct
  of $\sY$ restricted to $P^{\sY}$ equals the restriction of $\sB$ to the domain of
  $\Decode{\sX}$, or in other words to $P^{\sX}$. Since $\Decode{\sX}$ satisfies Definition~\ref{definition:newage}, and we only added relations outside $\rho$, so does $\sY$. Hence, the age of $\sY$ is contained in the age of $\coding{\sA}{\sB}$; by the homogeneity of $\coding{\sA}{\sB}$, it follows that $\sY$ embeds into $\coding{\sA}{\sB}$. Therefore, $\sX$ embeds into $
  \hrushovski{\sA}$.
\end{proof}

Note that being separated can be characterised by not containing the homomorphic image of any element of a finite set $\mathcal S$ of finite $\theta$-structures. 
As an immediate consequence of Proposition~\ref{prop:emb-forb} we therefore obtain the following corollary.

\begin{corollary}\label{cor:homobounded}
  Let $\sA$ be a $\rho$-structure with no algebraicity which is homomorphically bounded by a set $\mathcal{G}$ of $\rho$-structures.  Then 
  $\hrushovski{\sA}$ is homomorphically bounded by $\{\Completion{\sG} \mid \sG \in \mathcal{G}\} \cup
  \mathcal S$.
\end{corollary}


\section{Height 1 identities: local without global}
\label{section:local}
Let us recall that Question~(2) of the introduction asks whether the existence of a minion
homomorphism from $\Pol(\sA)$ to $\proj$ implies the existence of a uniformly
continuous minion homomorphism from $\Pol(\sA)$ to $\proj$. It has already been
established recently that there exists an $\omega$-categorical
structure with slow orbit growth which shows that the answer is negative~\cite{Bodirsky:2019aa}. However, that structure has an infinite number of relations and hence does not define a CSP, a fact that is inherent in its construction.

We are now going to prove that the Hrushovski-encoding of that structure, or in fact, of a simplification $\sS$ thereof, also provides an example. Since $\hrushovski{\sS}$ is a finite language structure, and since both $\omega$-categoricity and slow orbit growth are preserved by
the encoding, $\hrushovski{\sS}$ is a witness for the truth of
Theorem~\ref{thm:main:localnoglobalh1}. 

While the non-satisfaction of non-trivial h1~identities globally easily lifts from $\sS$ to $\hrushovski{\sS}$ by virtue of Proposition~\ref{prop:ucclonehomo}, we do not know in general when this is the case for the local satisfaction of 
non-trivial h1~identities.
Our proof thus relies on specific structural properties of $\sS$; we show that both $\sS$ and $\hrushovski{\sS}$ locally satisfy \emph{\dwnu\ identities}. This will also constitute an alternative proof of the fact that the original structure  $\sS$ satisfies non-trivial h1~identities locally -- the proof in~\cite{Bodirsky:2019aa} is indirect in the sense that it does not provide the actual identities satisfied in~$\sS$, a strategy which turned out infeasible for $\ES$.

\subsection{\Dwnu\ identities} \label{sect:phamidentities}
  We now define the \dwnu\ identities and argue
  that they are non-trivial. Following that, we prove Theorem~\ref{thm:embtoPham} 
  providing  a sufficient condition for the local satisfaction of such
  identities. 

\begin{definition}
  Let $n>k>1$, let $g_1,\ldots, g_{n}$ be binary function symbols,
  and for every injective function $\psi\colon \{1, \ldots, k\} \to \{1,
  \ldots, n\}$ let $f_\psi$ be a $k$-ary function symbol. Then the set of identities given by 
  \begin{align*}
    f_\psi(y,x,\ldots,x) &= g_{\psi(1)}(x,y)\\
    f_\psi(x,y,\ldots,x) &= g_{\psi(2)}(x,y)\\
    &\vdots\\
    f_\psi(x,\ldots,x,y) &= g_{\psi(k)}(x,y),
  \end{align*}
  for all injective functions $\psi\colon \{1, \ldots, k\} \to \{1, \ldots,
  n\}$ is called a set of \emph{\dwnu\ identities}, or simply
  \emph{DWNU identities} by abbreviation aficionados. In order to emphasise the
  parameters $n$ and $k$, we sometimes refer to the identities as \emph{$(n,
  k)$ \dwnu\ identities}.
\end{definition}


  Note that any function clone which satisfies identities of the form 
  $$
  f(y,x,\ldots,x) = \cdots =
  f(x,\ldots,x,y),
  $$ 
  called \emph{$k$-ary weak near-unanimity identities} when $f$ is $k$-ary for some $k
  \geq 3$, must also satisfy the $(n, k)$ \dwnu\
  identities for all $n > k$. This can be seen by setting
  $f_\psi = f$ for every $\psi$. Moreover,  there exist function clones which satisfy 
  \dwnu\ identities, but do not  satisfy any weak
  near-unanimity identities: one example is the clone generated by all injective
  operations on a countable set, see~\cite{Eq-oligo-CSP}. Hence, we can regard
  \dwnu\ identities as a strict weakening of the weak
  near-unanimity identities. 

  Further note that, for all parameters $m\geq n>k>1$, the $(n, k)$ \dwnu\ 
  identities form a subset of the $(m, k)$ \dwnu\ 
  identities. Thus for every fixed $k>1$ the family of $(n, k)$
  \dwnu\ identities form an infinite chain of h1
  identities of increasing strength. In the special case $k = 2$, the satisfaction of any of the $(k,2)$
  \dwnu\ identities is equivalent to the existence of a
  binary commutative term (as they imply $g_1(x,y) = g_2(y,x) = g_3(x,y) =
  g_1(y,x)$).

\begin{lemma} \label{lemma:phamnontrivial}
  For all parameters $n > k > 1$ the $(n,k)$ \dwnu\
 identities are non-trivial.
\end{lemma}

\begin{proof}
Assume to the contrary that there exist projections
  $g_1,\ldots, g_n \in \Projs$ and $f_\psi \in \Projs$ for every injection
  $\psi\colon \{1, \ldots, k\} \to \{1, \ldots, n\}$ that satisfy the $(n,k)$
  \dwnu\ identities. First, suppose that there are two
  distinct $1\leq i,j\leq k$ such that $g_i, g_j$ are both the projection onto the second coordinate. Then let $\psi$
  be an injective function with $\psi(1)= i, \psi(2) = j$. It follows from the
  identities that $f_\psi(y,x,\ldots,x) = f_\psi(x,y,\ldots,x) = y$ holds for all values of the variables,
  which contradicts $f_\psi$ being a projection. Therefore at most one
  operation $g_i$ equals the projection to its second coordinate. Since $n
  > k$, there is an injective function $\psi\colon \{1, \ldots, k \} \to \{1,
  \ldots, n\}$ such that $g_{\psi(i)}$ is the first projection for all $i \in \{1, \ldots, k\}$.
  Then $f_\psi$ satisfies the weak near-unanimity identities, which again
  contradicts $f_\psi$ being a projection.
\end{proof}

We now prove Theorem~\ref{thm:embtoPham}. Recall that
the theorem states that a homogeneous structure $\sU$ satisfies $(n, k)$ \dwnu\ identities on a finite subset $F$ of its domain if the following two
assumptions hold:
  \begin{enumerate}[label = \textrm{ (\roman*)}]
    \item Only relations of arity smaller than $k$ hold on $F$;
    \item There is an embedding from $\sU^2$ into $\sU$.
  \end{enumerate}

Before we prove Theorem~\ref{thm:embtoPham}, observe that condition (ii) is equivalent to the existence of embeddings from arbitrary powers of $\sU$ into $\sU$.

\begin{lemma}\label{lemma:embsq}
  Let $\sU$ be a relational structure and let $n \geq 2$. Then there exists an
  embedding from $\sU^2$ into $\sU$ if and only if there exists an embedding from $\sU^n$ into $\sU$.
\end{lemma}

\begin{proof}
If there is an embedding $f \colon \sU^n \to \sU$ for some $n  \geq 2$, then 
$g \colon \sU^2 \to \sU$, defined by $g(x, y) := f(x, y,
\ldots, y)$, is also an embedding. On the other hand, if for some $n \geq 2$ there
exist embeddings $g \colon \sU^2 \to \sU$ and $h \colon \sU^n \to \sU$, then the
composition $f(x_1, \ldots, x_{n + 1}) = g(h(x_1, \ldots, x_n), x_{n + 1})$ is
an embedding from $\sU^{n+1}$ into $\sU$. Hence by induction the existence of an embedding from $\sU^2$ into $\sU$
implies the existence of an embedding from $\sU^n$ into $\sU$ for all $n \geq 2$.
\end{proof}

\begin{proof}[Proof of Theorem~\ref{thm:embtoPham}]
  For all $l\geq 2$, define $X_l \subseteq F^l$ by
  \[
    X_l = \bigcup_{a,b\in F} \{(a,\ldots,a,b),(a,\ldots,a,b,a), \ldots,(b,a,\ldots,a) \},
  \] 
  and let $\sX_l$ be the substructure which $X_l$ induces in $\sU^l$. 
  
  The first
  step of our proof is to show that if $n \geq k$, then there exists an embedding
  $h \colon \sX_k \to \sX_{n}$ such that $\mathbf{x}$ is an initial segment of $h(\mathbf{x})$ for all $\mathbf{x}\in X_k$.
 Let us first assume that $k \geq 3$. For every tuple $\mathbf{x} \in \sX_k$ we are then going to denote the unique element of $F$ which
occurs more than once among its entries by $s(\mathbf{x})$. Define $h
\colon X_k \to X_{n}$ to be the map that extends the tuple $\mathbf{x}$ by $n-k$ many
entries with value $s(\mathbf{x})$. In order to prove that $h$ is an embedding let
$\mathbf{x}_1, \ldots, \mathbf{x}_m \in X_k$ be such that $R^{\sU^k}(\mathbf{x}_1,
\ldots, \mathbf{x}_m)$ holds for some $m$-ary relation symbol $R$ in the signature of $\sU$.
By assumption~(i) we have $m<k$. Thus
there exists $1\leq j\leq k$ such that the projection of each
$\mathbf{x}_i$ to its $j$-th coordinate equals $s(\mathbf{x}_i)$.
Therefore $(s(\mathbf{x}_1),\ldots,s(\mathbf{x}_m)) \in R^{\sU}$, and hence $h$ is a homomorphism. Its inverse -- the projection of $n$-tuples to the
first $k$-coordinates -- is also a homomorphism, and thus $h$ is an embedding. In the remaining case of $k = 2$, we define $h(x_1,x_2)= (x_1,
x_2,\ldots,x_2)$. To check that this $h$ is an embedding, by assumption (i), we only need to check that $h$ is an
embedding with respect to unary relations, which however follows from its definition.

Observe that $h$ was defined in such a way that, for each index $1\leq i\leq k$, the $i$-th
projection of $h(\mathbf{x})$ is equal to $x_i$. By permuting the
coordinates of its image in a suitable manner, we can obtain embeddings
$h_\psi\colon \sX_k \to \sX_{n}$ for every injection $\psi \colon \{1, \ldots,
k\} \to \{1, \ldots, n\}$ such that the $\psi(i)$-th projection of
$h_\psi(\mathbf{x})$ is equal to $x_i$ for all $1\leq i\leq k$.

  In order to construct the operations $f_\psi$ on $F$, let $f \colon \sU^k \to
  \sU$ and $g \colon \sU^n \to \sU$ be embeddings, which exist by
  Lemma~\ref{lemma:embsq}. For every injection $\psi:\{1, \ldots, k\} \to \{1,
  \ldots, n\}$ define the map $u_\psi\colon f(\sX_k) \to g(\sX_{n})$ by 
  \begin{align} \label{eq:localpham}
    u_\psi(f(a, \ldots, a, \underset{i^{\text{th}}}{b}, a, \ldots, a)) = g(a,
    \ldots, a, \underset{\psi(i)^{\text{th}}}{b}, a, \ldots, a).
  \end{align}
  Then $u_\psi$ is equal to $g \circ h_\psi \circ f^{-1}$. Since $h_\psi$ is an
  embedding, $u_\psi\colon f(\sX_k) \to u_\psi (f(\sX_k))$ is an isomorphism
  between finite substructures of $\sU$. By the homogeneity of $\sU$, it can be
  extended to an automorphism $v_\psi$ of $\sU$. Set $f_\psi := v_\psi \circ f$
  and, for all $1\leq i\leq n$, define $g_i(x, y) := g(x,\ldots,x,y,x,\ldots,x)$, where the only $y$ appears at
  the $i$-th coordinate of $g$. It then follows from \eqref{eq:localpham} 
  that these polymorphisms satisfy the $(n,k)$ \dwnu\
  identities on $F$, concluding the proof.
\end{proof}

\subsection{Revisiting the infinite language
counterexample}\label{section:infinitelang}
We now revisit the infinite language structure presented in~\cite{Bodirsky:2019aa} which provides a negative answer to Question~(2). In fact, the construction there depends on two parameters $\alpha$ and $\delta$, of which only $\alpha$ is mentioned, whereas $\delta$ is eliminated by an (arbitrary) choice. Therefore, actually a family of structures are presented, which will be of importance to us when we study the CSPs of those structures in Section~\ref{section:complexity}, which depends on the parameters $\alpha$ and $\delta$.

We are going to recall the construction of the structures, or in fact a slight simplification thereof, as we do not require
  them to be model-complete cores. This additional condition was necessary in~\cite{Bodirsky:2019aa} because of the indirect proof of the local satisfaction of non-trivial h1~identities; since we are going to prove directly the satisfaction of \dwnu\ 
  identities, we can avoid these technicalities.

%

\begin{theorem}[\cite{CherlinShelahShi}, Corollary of Theorem~3.1 in~\cite{Hubicka2009UniversalSW}]\label{thm:CSS}
  Let $\mathcal{F}$ be a finite family of finite connected relational
  structures. Then there exists a countable $\omega$-categorical structure
  $\CSS(\mathcal{F})$ such that
  \begin{itemize}
    \item $\CSS(\mathcal{F})$ is homomorphically bounded by $\mathcal{F}$;
    \item $\CSS(\mathcal{F})$ has no algebraicity; 
    \item there exists a homogeneous expansion $\sH$ of\/ $\CSS(\mathcal{F})$ by finitely many pp-definable relations whose arities are the size of the minimal cuts of structures in $\mathcal{F}$,
    and\/ $\sH$ is homomorphically bounded.
  \end{itemize}
\end{theorem}
We refer to~\cite{Hubicka2009UniversalSW} for further details,  including the definitions of connectedness and cuts. 

The first step in the construction is to use Theorem~\ref{thm:CSS} to obtain $\omega$-categorical structures that are homomorphically bounded by a given connected graph on $n$-tuples. More precisely, for every finite connected loopless graph $\mathbb{G}$ and every integer $n \geq 1$ define $\mathbb{G}[n]$ to be a structure with a $2n$-ary predicate $R$; the domain of $\mathbb{G}[n]$ is obtained by substituting every vertex $x$ of $\mathbb{G}$ by $n$ distinct elements $x_1,\ldots,x_n$, and the relation $R^{\mathbb{G}[n]}$ is defined to contain all tuples $(x_1,\ldots,x_n,y_1,\ldots,y_n)$ for which $(x,y)$ is an edge in $\mathbb{G}$. Furthermore let ${\mathbb L}_1^{2n},\ldots, {\mathbb L}_{N}^{2n}$ be all the `loop-like' $R$-structures, that is all structures of size $2n-1$ in which $R$ holds for precisely one $2n$-tuple.

For every finite connected loopless graph $\mathbb{G}$, let $\sS(\mathbb{G},n)$ be the structure obtained from Theorem~\ref{thm:CSS}
for the set $\mathcal F:=\{\mathbb{G}[n],{\mathbb L}_1^{2n},\ldots, {\mathbb L}_{N}^{2n}\}$, and let $\sH(\mathbb{G},n)$ be its homogeneous homomorphically bounded expansion
whose existence is claimed in Theorem~\ref{thm:CSS}. The following statement follows from the results in~\cite{Bodirsky:2019aa} (although not explicitly stated there, it can be inferred from the proof of Lemma~6.5).

\begin{lemma} \label{lemma:CSSconstruction} 
  Let $\mathbb{G}$ be a finite connected non-trivial loopless graph.
   All tuples related by a relation in $\sH(\mathbb G, n)$ have at least $n$ distinct entries.
\end{lemma}

In the next step, we superpose structures of the form $\sS(\mathbb G, n)$ and of the form $\sH(\mathbb G, n)$, respectively,  in a generic way to obtain in a single structure, as
in~\cite[Section~6.2]{Bodirsky:2019aa}. Suppose that $\sA$ and $\sB$ are two
 structures with no algebraicity, and without loss of
generality assume that their signatures $\sigma$ and $\tau$ are disjoint. Then
their \emph{generic superposition} $\sA \odot \sB$ is defined in the following way.
\begin{itemize} 
  \item Let $\sA'$ and $\sB'$ be homogeneous first-order expansions of $\sA$ and $\sB$ in disjoint signatures $\sigma'$ and $\tau'$. Note that both the age of $\sA'$ and the age of  $\sB'$ have the SAP.
  \item Let $\mathcal{C}$ be the class of finite $(\sigma' \cup \tau')$-structures such that their $\sigma'$- and $\tau'$-reducts embed into $\sA'$ and $\sB'$ respectively. Then $\mathcal{C}$ is also a Fra\"{i}ss\'{e} class, and in fact it also has SAP. We then define $\sA \odot \sB$ to be the $(\sigma \cup \tau)$-reduct of the Fra\"{i}ss\'{e} limit of $\mathcal{C}$.
\end{itemize}
In a similar fashion we can also form the generic superposition of a family of
countably many structures with no algebraicity. 





\begin{definition}\label{definition:S}
Let $\alpha \colon \N\setminus\{0\} \to \N$ be a strictly monotone map and let $\delta$ be a map from $\N\setminus\{0\}$ to the set of loopless connected graphs  which contains all non-3-colourable graphs in its image. Then we define
  $\sS_{\delta, \alpha}$ and $\sH_{\delta, \alpha}$ to be the generic
  superpositions of the families $(\sS(\delta(n), \alpha(n)))_{n\geq 1}$ and $(\sH(\delta(n),
  \alpha(n)))_{n\geq 1}$ respectively. 
\end{definition}

The superposed structures have the following properties.

\begin{proposition}\label{proposition:SandH}
  Let $\delta$ and $\alpha$ be as in Definition~\ref{definition:S}. Then the following statements hold.
  \begin{enumerate}[label = \textrm{ (\arabic*)}]
     \item $\sH_{\delta, \alpha}$ is a homogeneous first-order expansion of $\sS_{\delta,
       \alpha}$ by pp-definable relations; 
     \item $\sS_{\delta, \alpha}$ (and hence also $\sH_{\delta, \alpha}$) is  
       $\omega$-categorical and has no algebraicity;
     \item $\sS_{\delta, \alpha}$ and $\sH_{\delta, \alpha}$ are homomorphically bounded; 
     \item There exists a minion homomorphism from $\Pol(\sS_{\delta, \alpha})$ (and hence also from $\Pol(\sH_{\delta, \alpha})$)
       to $\proj$.
  \end{enumerate}
\end{proposition}

\begin{proof} 
  For~(1), note that each $\sH_{\delta, \alpha}$ is homogeneous by the construction of the superposition. Any relation of $\sH_{\delta, \alpha}$ is a relation of $\sH(\delta(n),
  \alpha(n))$ for some $n\geq 1$. Thus, it is first-order definable in $\sS(\delta(n),
  \alpha(n))$, and hence also in $\sS_{\delta, \alpha}$. Item~(2) can be proven as in~\cite[Lemma~6.5]{Bodirsky:2019aa}. To see~(3), note that $\sS(\delta(n),
  \alpha(n))$ is homomorphically bounded by a set $\mathcal F_n$ for all $n\geq 1$; taking all possible expansions of all structures from $\bigcup_{n\geq 1} \mathcal F_n$ to the signature of $\sS_{\delta, \alpha}$ yields a set by which $\sS_{\delta, \alpha}$ is homomorphically bounded. The same argument works for $\sH_{\delta, \alpha}$.  Item~(4) can be shown by the same proof as in~\cite[Lemma~6.7]{Bodirsky:2019aa}.
\end{proof}

By Proposition~\ref{proposition:SandH} (3), the structure $\sH_{\delta,
\alpha}$ is homomorphically bounded; therefore it satisfies the condition of
the following lemma.

\begin{lemma} \label{lemma:powerembedding}
  Let $\sA$ be a homogeneous homomorphically bounded structure and let $k \geq 1$. Then there exists an embedding from $\sA^k$ into $\sA$.
\end{lemma}

\begin{proof}
  Let $\sA$ be homomorphically bounded by $\mathcal F$. We first claim that no structure from $\mathcal F$ homomorphically maps into $\sA^k$. Suppose for a contradiction that there exists $\mathbb{X}\in \mathcal F$ and a homomorphism $h: \mathbb{X} \to \sA^k$. Composing $h$
  with the projection of $\sA^k$ to the first coordinate, we obtain a
  homomorphism from $\mathbb X$ to $\sA$, which is a contradiction. Hence, the age of $\sA^k$ is contained in the age of $\sA$. By the homogeneity of $\sA$, a standard argument shows that $\sA^k$ embeds into $\sA$.
\end{proof}

\subsection{The finite language counterexample} \label{section:liftingtofin}

We are now ready to prove that the
  Hrushovski-encoding $\hrushovski{{\sS_{\delta, \alpha}}}$ of $\sS_{\delta,
  \alpha}$ satisfies \dwnu\ identities locally, and
  therefore has no uniformly continuous minion homomorphism to $\Projs$.
  Note that $\hrushovski{{\sS_{\delta, \alpha}}}$ is well-defined since
  $\sS_{\delta, \alpha}$ has at most one relation in every arity and no
  algebraicity by Proposition~\ref{proposition:SandH}.



\begin{theorem}\label{thm:localTrungterms2}
  Let $\delta$ and $\alpha$ be as in Definition~\ref{definition:S}, and let $F$ be a finite subset of the domain of 
  $\hrushovski{\sS_{\delta, \alpha}}$. Then there exists $k > 1$ such that
  $\Pol(\hrushovski{\sS_{\delta, \alpha}})$ satisfies the $(n, k)$ \dwnu\ 
  identities on $F$ for all $n > k$. 
\end{theorem}

\begin{proof}
  For the sake of notational lightness, denote by $\sB$ the homogeneous first-order
  expansion $\sH_{\delta, \alpha}$ of $\sS_{\delta,\alpha}$. Let $\rho$ and
  $\sigma$ be the signatures of $\sS_{\delta, \alpha}$ and of $\sB$, respectively. The Hrushovski-encoding
  $\hrushovski{\sS_{\delta, \alpha}}$ is then a reduct of the blowup $\coding{\sS_{\delta, \alpha}}{\sB}$, and hence $\Pol(\coding{\sS_{\delta,
  \alpha}}{\sB}) \subseteq \Pol(\hrushovski{{\sS_{\delta, \alpha}}})$. We claim
  that there exists some $k> 1$ for which $\coding{\sS_{\delta, \alpha}}{\sB}$
  satisfies the $(n,k)$ \dwnu\ identities on $F$, in
  which case $\hrushovski{{\sS_{\delta, \alpha}}}$ satisfies these identities on
  $F$ as well.

  In order to prove the claim we verify that conditions (i) and (ii) of
  Theorem~\ref{thm:embtoPham} hold for $\coding{\sS_{\delta, \alpha}}{\sB}$, F, and
  a suitable $k > 1$. By Proposition~\ref{proposition:SandH}, the structure 
  $\sB = \sH_{\delta,\alpha}$ is homomorphically bounded, so
  Lemma~\ref{lemma:powerembedding} implies that $\sB^2$ embeds
  into $\sB$. By Lemma~\ref{lemma:encodingextensions}~(4), there
  exists an embedding of $(\coding{\sS_{\delta, \alpha}}{\sB})^2$ into
  $\coding{\sS_{\delta, \alpha}}{\sB}$, and thus condition (ii) holds.
	
  It remains to check~(i) which states that there exists an upper bound on
  the arity of tuples in $F$ that satisfy some relation from
  $\coding{\sS_{\delta, \alpha}}{\sB}$.   Suppose that
  $R^{\coding{\sS_{\delta, \alpha}}{\sB}}$ contains a tuple entirely within $F$ for some $R\in\sigma^+$, the language of $\coding{\sS_{\delta, \alpha}}{\sB}$. Since $\sigma^+=\sigma\cup\theta$, and all relations in $\theta$ have arity at most $4$, we may assume that $R \in \sigma$. Then any tuple in $R^{\coding{\sS_{\delta, \alpha}}{\sB}}$ must lie entirely within $P^{\coding{\sS_{\delta, \alpha}}{\sB}}$, and so the tuple is an element of $R^{\sB}$, by Lemma~\ref{lemma:A+hasA}. 
Since we constructed $\sB=\sH_{\delta, \alpha}$ as the superposition of
  the family $(\sH(\delta(n), \alpha(n)))_{n\geq 1}$, the symbol $R$ lies in the signature of  $\sH(\delta(n), \alpha(n))$ for some $n\geq 1$. By
  Lemma~\ref{lemma:CSSconstruction}, at least $\alpha(n)$ many of the values of any tuple in   $R^{\sB}$ are distinct. Therefore, $\alpha(n)$ must be smaller
  than $|F|$. Since $\alpha$ is a strictly increasing function and each $\sH(\delta(n),
  \alpha(n))$ has a finite language, it follows that only finitely many relations
  of $\sB= \sH_{\delta, \alpha}$ have tuples that lie entirely in $F$. Let $k> 1$ be a
  strict upper bound on the arity of those relations. For this choice of $k$ we have that~(ii) of of Theorem~\ref{thm:embtoPham} holds, and thus $\coding{\sS_{\delta, \alpha}}{\sB}$ satisfies the $(n,k)$ \dwnu\  
  identities on $F$ for all $n > k$.
  \end{proof}

It follows that the original structures $\sS_{\delta, \alpha}$ satisfy \dwnu\ identities locally as well, since by Proposition~\ref{prop:ucclonehomo}, there is a uniformly continuous minion homomorphism from $\Pol(\hrushovski{\sS_{\delta, \alpha}})$ to $\Pol(\sS_{\delta, \alpha})$.
	This result is new and no other explicit description of non-trivial local h1~identities of $\sS_{\delta, \alpha}$ was given in~\cite{Bodirsky:2019aa}.

We are now ready to prove Theorem~\ref{thm:main:localnoglobalh1}. 

\begin{proof}[Proof of Theorem~\ref{thm:main:localnoglobalh1}]
  It was shown in Lemma \cite[Lemma 6.6]{Bodirsky:2019aa} that there are
  choices of the functions $\alpha$ and $\delta$ (as in Definition~\ref{definition:S}) such that $\sS_{\delta,
  \alpha}$ is not only $\omega$-categorical, but it also has slow orbit growth; this is the case if $\alpha$ grows sufficiently fast. 
  We will show that any such $\hrushovski{{\sS_{\delta, \alpha}}}$ satisfies the properties of the required
  $\sS$. Note that $\hrushovski{{\sS_{\delta, \alpha}}}$ has a finite relational
  signature. 
  
By Theorem~\ref{thm:localTrungterms2}, for every finite
  subset $F$ of $\hrushovski{{\sS_{\delta, \alpha}}}$ the clone
  $\Pol(\hrushovski{{\sS_{\delta, \alpha}}})$ satisfies some \dwnu\ 
  identities on $F$. By Lemma~\ref{lemma:phamnontrivial}, the
  identities are non-trivial, and hence there is no uniformly continuous minion
  homomorphism from $\Pol(\hrushovski{{\sS_{\delta, \alpha}}})$ to $\Projs$.

  Finally, by Proposition~\ref{prop:ucclonehomo} we have that 
  $\Pol(\hrushovski{{\sS_{\delta, \alpha}}})$ has a clone homomorphism to
  $\Pol({\sS_{\delta, \alpha}})$.  There exists a minion homomorphism from
  $\Pol(\sS_{\delta, \alpha})$ to $\Projs$ by
  Proposition~\ref{proposition:SandH}~(4), so the composition of the two homomorphisms gives us a
  minion homomorphism from $\Pol(\hrushovski{{\sS_{\delta, \alpha}}})$ to
  $\Projs$, which completes the proof. 
\end{proof}

\section{A Hierarchy of Hard Constraint Satisfaction Problems}
\label{section:complexity}
Next, we investigate the complexity of $\csp$s of structures encoded by the Hrushovski-encoding. We will mostly encode \emph{trivial} structures, that is, structures whose relations are all empty (but whose signature might be complex). In Section~\ref{S:CoderLangage} we show that for every language $L$ we can construct a trivial structure $\sT$ such that $L$ reduces to $\csp(\hrushovski{\sT})$ in logarithmic space, and such that there is a so called \emph{$\coNP$-many-one reduction} from $\csp(\hrushovski{\sT})$ to $L$. This implies the completeness result in Theorem~\ref{thm:complexity-main}.  In Sections~\ref{S:coNP-complete} and~\ref{S:Intermediaire} we
perform a more detailed analysis for the case where $L \in \PolComplexity$ and show in particular that we can obtain \coNP-intermediate CSPs (assuming that $\PolComplexity\ \neq \coNP$). In Section \ref{section:pseudo} we use encodings of trivial structures to prove Theorem \ref{thm:pseudo}.

\subsection{Encoding arbitrary languages}
\label{S:CoderLangage}
We begin by giving a formal definition of trivial structures and edge
structures.

\begin{definition}\label{D:ERS}
  For an alphabet $\Sigma$ and a language $W\subseteq\Sigma^{\geq 2}$, let $\rho_W$ be
  the signature consisting of $|\ww|$-ary relation symbols $R_\ww$ for every
  word $\ww \in W$. The \emph{trivial structure} $\sT_W$ is the countable
  $\rho_W$-structure with all relations empty.

  For every word $\ww\in W$, the \emph{$\ww$-edge structure $\sF_w$} is the
  $\rho_W$-structure on the set $F_\ww=\set{1,\dots,|w|}$ whose only non-empty relation is
  $R_\ww^{\sF_\ww}=\set{(1,\dots,|w|)}$.
\end{definition}

The trivial structure $\sT_W$ is homomorphically bounded by the set of
all edge-structures $\sF_\ww$ with $\ww\in W$. Moreover, $\sT_W$ has  no algebraicity. In the following lemma we show that trivial structures and their encodings have the
algebraic properties required in Theorem~\ref{thm:complexity-main}. 

\begin{lemma} \label{L:trivialprop}
  Let $\sT_W$ be the trivial structure for some $W\subseteq\Sigma^{\geq 2}$. Then both
  $\sT_W$ and $\hrushovski{\sT_W}$ are $\omega$-categorical, have slow orbit growth, and satisfy non-trivial h1~identities. Furthermore $\hrushovski{\sT_W}$ is homogeneous in a finite language.
\end{lemma}

\begin{proof}
  It follows immediately from the definition that $\sT_W$ is both $\omega$-categorical and has slow
  orbit growth. By Proposition~\ref{prop:omegacat} its encoding
  $\hrushovski{\sT_W}$ is also $\omega$-categorical of slow orbit growth. It further easy to see that $\hrushovski{\sT_W}$ is homogeneous.

  In order to show that the structures satisfy some non-trivial h1~identities,
  note that $\sT_W^2$ embeds into $\sT_W$. Let $\sB = \coding{\sT_W}{\sT_W}$ be
  the blow-up of $\sT_W$. By
  Lemma~\ref{lemma:encodingextensions}~(4), $\sB^2$ embeds into $\sB$.
  Moreover, the non-empty relations of $\sB$ are of arity at most $4$. Hence,
  by Theorem~\ref{thm:embtoPham}, $\sB$ satisfies $(6,5)$ dissected weak
  near-unanimity identities locally. By a standard compactness argument we
  obtain that $\sB$ satisfies $(6,5)$ dissected weak near-unanimity identities
  globally. Since $\hrushovski{\sT_W}$ is a reduct of $\sB$,
  $\hrushovski{\sT_W}$ also satisfies the identities. It follows that $\sT_W$ satisfies the same non-trivial h1 identities. 
\end{proof}

Since $\sT_W$ is homomorphically bounded by the edge structures $\setm{\sF_\ww}{\ww\in W}$,
Proposition~\ref{prop:emb-forb} can be used to give an explicit description of
$\csp(\hrushovski{\sT_W})$.

\begin{lemma}\label{L:DescriptionCSP}
  Let $W\subseteq\Sigma^{\geq 2}$, and let $\sX$ be a $\theta$-structure. Then the
  following are equivalent.
  \begin{enumerate}[label = $(\arabic*)$]
    \item There exists a homomorphism from $\sX$ to $\Hrushovski{\sT_W}$;
		
    \item $\sX$ is separated and there is no word $\ww\in W$ 
            such that $\Completion{\sF_\ww}$ homomorphically
            maps to $\sX$.
            
    \item $\sX$ is separated and there is no word $\ww\in W$ of length smaller
            than $\card{X}$ such that $\Completion{\sF_\ww}$ homomorphically
            maps to $\sX$.
\end{enumerate}
\end{lemma}

\begin{proof}
  The equivalence of (1) and (2) follows from Proposition~\ref{prop:emb-forb}.
  To demonstrate the equivalence of (2) and (3), observe that if there is a
  homomorphism from $\Completion{\sF_\ww}$ to $\sX$, then $\sX$ contains a valid
  $\ww$-code, and so $\card{\ww} \leq \card{X}$.
\end{proof}

To prove our complexity results,  it will be convenient to use the notion of a \emph{$\coNP$-many-one reduction}. Such reductions were first defined by Beigel, Chang, and Ogiwara in \cite{BCO}; we are going to use the following equivalent definition.

\begin{definition}
  Let $K$ and $L$ be two languages in an alphabet $\Sigma$. Then a \emph{$\coNP$-many-one reduction} from
  $K$ to $L$ is a non-deterministic Turing Machine $M$ such that $M$ runs in polynomial
  time, and for all words $\ww$ over $\Sigma$ we have $\ww\in K$ if and only if each
  path of $M$, on input $\ww$, computes a word in $L$.
\end{definition}

Note that having a $\coNP$-many-one reduction from $K$ to $L$ is a stronger condition than $K$ being in $\coNP^L$ (i.e. having a \emph{Turing} $\coNP$-reduction from $K$ to $L$), since then there are no restrictions on when and how to use the oracle $L$. If for instance $K$ has a $\coNP$-many-one reduction to a problem that is in $\coNP$, $K$ is also in $\coNP$ -- but this is not necessarily true for Turing $\coNP$-reductions.

The following lemma generalizes this fact. It follows from the easily verified fact that the composition of a Turing $\coNP$-reduction and a $\coNP$-many-one reduction is a Turing $\coNP$-reduction.


\begin{lemma}\label{L:CompositionReduction}
  For every complexity class $\mathcal C$ the class $\coNP^\mathcal C$ is
  closed under $\coNP$-many-one reductions.
\end{lemma}

We are now ready to encode arbitrary languages as CSPs of Hrushovski-encoded structures.

\begin{theorem}\label{T:CSPcodeLangage}
  Let $L\subseteq \Sigma^{\geq 2}$ be a language such that both $L$ and its complement $W = \Sigma^{\geq 2} \setminus L$ are non-empty. Then $L$ has a
  log-space many-one reduction to $\csp(\hrushovski{\sT_W})$, and
  $\csp(\hrushovski{\sT_W})$ has a $\coNP$-many-one reduction to $L$.
\end{theorem}

\begin{proof}
It is easy to see that the function $\ww \mapsto \Completion{\sF_\ww}$ is computable in logarithmic space with respect to
  $\card{\ww}$. Also note that there is a homomorphism $\Completion{\sF_\wu}
  \to \Completion{\sF_\ww}$ if and only if $\ww = \wu$. Moreover, it follows
  from Lemma~\ref{L:DescriptionCSP} applied to $\sX = \Completion{\sF_\ww}$ 
  that there is a homomorphism $\Completion{\sF_\ww} \to \hrushovski{\sT_W}$ if
  and only if $\ww \in L$. Thus $L$ has a log-space many-one reduction to
  $\csp(\hrushovski{\sT_W})$.
  
  For the other reduction, let $\sX$ be a finite $\theta$-structure, an instance of $\csp(\hrushovski{\sT_W})$. If there is no homomorphism $\sX \to  \hrushovski{\sT_W}$, by
Lemma~\ref{L:DescriptionCSP}, either $\sX$ is not separated (which can be
  checked in polynomial time), or there is a word $\ww \in W$ not longer than the size of the domain of $\sX$ and a homomorphism $f \colon \Completion{\sF_\ww} \to \sX$.
  The reduction does the following: if $\sX$ is not separated, we map it to a fixed element of $W$. 
  Otherwise, we guess a word $\ww$ not longer than the size of the domain of $\sX$ and a function $f\colon \Completion{\sF_\ww}\to\sX$.
  If this function is not a homomorphism, we map $\sX$ to a fixed word of $L$.
  If $f$ is a homomorphism, we map $\sX$ to $\ww$.
  Thus, if $\sX\in\csp(\hrushovski{\sT_W})$ then all runs of the reduction output a word of $L$.
  Moreover, if $\sX\notin\csp(\hrushovski{\sT_W})$, then at least one run outputs word in $W$.
\end{proof}

As a direct consequence of Theorem~\ref{T:CSPcodeLangage} we obtain the completeness result in 
Theorem \ref{thm:complexity-main}.

\begin{corollary} \label{cor:CSPcompletness}
  Let $\mathcal C$ be a complexity class such that there exist $\coNP^\mathcal
  C$-complete problems. Then there exists $W \subseteq \{0,1\}^{\geq 2}$ such that
  $\csp(\Hrushovski{\sT_{W}})$ is $\coNP^\mathcal C$-complete. In particular, we
  have complete problems of the form $\csp(\Hrushovski{\sT_{W}})$ for the
  following classes:
\begin{itemize}
\item $\Pi_n^{\PolComplexity}$ -- part of the polynomial hierarchy; 
\item \PSPACE; 
\item \EXPTIME; 
\item the fast-growing time complexity classes $\mathbf{F_\alpha}$ where
        $\alpha\geq 2$ is an ordinal (such as the classes \textsc{Tower},
        \textsc{Ackermann}, and \textsc{Hyperackermann},
        see~\cite{DBLP:journals/toct/Schmitz16}).
\end{itemize}
\end{corollary}

\begin{proof}
  Let $L \subseteq \{0,1\}^{\geq 2}$ be a $\coNP^\mathcal C$-complete language, and
  let $W$ be its complement. Then $L$ reduces to $\csp(\Hrushovski{\sT_{W}})$
  by Theorem~\ref{T:CSPcodeLangage}, and so $\csp(\Hrushovski{\sT_{W}})$ is
  $\coNP^\mathcal C$-hard. On the other hand, there is a $\coNP$-many-one reduction of $\csp(\Hrushovski{\sT_{W}})$ to $L$. Thus, by Lemma~\ref{L:CompositionReduction}
  $\csp(\Hrushovski{\sT_{W}})$ belongs to $\coNP^\mathcal C$.
\end{proof}

We remark that in Corollary~\ref{cor:CSPcompletness} we used encodings with respect to an alphabet
$\Sigma$ that is not unary (in order words, we used our refinement of the original encoding due to Hrushovski). This is indeed necessary, as for instance the existence of unary \PSPACE-hard language would imply \PolComplexity = \PSPACE.

If the language $L$ in Theorem~\ref{T:CSPcodeLangage} is undecidable, then
$\csp(\Hrushovski{\sT_W})$ is undecidable of the same Turing degree. Thus we
obtain the following additional corollary.

\begin{corollary} \label{cor:CSPundecidability}
  For every undecidable Turing degree $\tau$ there exists a set $W \subseteq
  \Sigma^{\geq 2}$ such that $\csp(\hrushovski{\sT_W})$ undecidable of degree $\tau$.
\end{corollary}

\subsection{\coNP-complete $\csp$s} \label{S:coNP-complete}
Observe that if a language $L$ is in $\PolComplexity$, then, by
Theorem~\ref{T:CSPcodeLangage}, $\csp(\hrushovski{\sT_L})$ is in $\coNP$. In
this section, we consider two special cases -- $L$ being finite and cofinite.
In the first case, $\csp(\hrushovski{\sT_L})$ is $\PolComplexity$ and, in the
second case, it is $\coNP$-complete. These results are used in the next
section to obtain a $\coNP$-intermediate  CSP.

\begin{lemma}\label{L:EW-Reduction}
  Let $V\subseteq W\subseteq\Sigma^{\geq 2}$. If $W\setminus V$ is finite, then there
  is a polynomial-time reduction from $\csp(\hrushovski{\sT_W})$ to
  $\csp(\hrushovski{\sT_V})$.
\end{lemma}

\begin{proof}
  Let $w_1,\dots,w_n$ be the elements of $W\setminus V$. Denote by $N$ the
  maximal length of a word in $\{w_1,\dots,w_n\}$.  It follows from Lemma~\ref{L:DescriptionCSP} that,
  for a $\theta$-structure $\sX$, there is a homomorphism $\sX\to
  \hrushovski{\sT_W}$ if and only if $\sX$ is separated and for all $\ww\in W$
  there is no homomorphism $\Completion{\sF_{\ww}} \to \sX$. There is a similar
  characterisation for the existence of a homomorphism $\sX\to
  \hrushovski{\sT_V}$. Therefore, there is a homomorphism $\sX\to
  \hrushovski{\sT_W}$ if and only if there is a homomorphism $\sX\to
  \hrushovski{\sT_V}$ and for all $1\leq i\le n$ there is no homomorphism
  $\Completion{\sF_{w_i}} \to \sX$.

  Finally, given $\sX$, computing whether there is a homomorphism
  $\Completion{\sF_{w_i}} \to \sX$ for some $1\leq i \leq n$, can be done in time
  $\bigO(\card{X}^{2N})$. Hence there is a polynomial-time reduction from
  $\csp(\hrushovski{\sT_W})$ to $\csp(\hrushovski{\sT_V})$.
\end{proof}

\begin{corollary}\label{C:EW-Facile}
  Let $W\subseteq\Sigma^{\geq 2}$ be finite. Then $\csp(\hrushovski{\sT_W})$ is
  solvable in polynomial time. 
\end{corollary}

\begin{proof}  
  It follows from Lemma~\ref{L:DescriptionCSP} that a $\theta$-structure $\sX$
  has a homomorphism to $\hrushovski{\sT_\emptyset}$ if and only if it is
  separated. This can be determined in polynomial time, and so
  $\csp(\hrushovski{\sT_W})$ is in $\PolComplexity$ by
  Lemma~\ref{L:EW-Reduction}.
\end{proof}

\begin{theorem}\label{T:CSPdifficile}
  Let $W\subseteq\Sigma^{\geq2}$ be such that $\Sigma^{\geq2}\setminus W$ is finite. Then
  $\csp(\hrushovski{\sT_{W}})$ is \coNP-complete.
\end{theorem}

\begin{proof}
  By Lemma~\ref{L:EW-Reduction}, it suffices to prove the theorem for
  $W=\Sigma^{\geq2}$. Moreover, it follows from Theorem~\ref{T:CSPcodeLangage} that
  $\csp(\hrushovski{\sT_{\Sigma^{\geq2}}})$ is in $\coNP$. We are going to reduce the
  clique problem, which is known to be \NP-complete, to the complement of
  $\csp(\hrushovski{\sT_{\Sigma^{\geq2}}})$. 

  Let $\sG=(V,E)$ be a finite loopless graph and let $n\ge 2$ be an integer. By
  Lemma~\ref{L:DescriptionCSP}, there is a homomorphism
  $\sX\to\hrushovski{\sT_{\Sigma^{\geq2}}}$ for some $\theta$-structure $\sX$ if and
  only if $\sX$ is separated and for all $\ww\in\Sigma^{\geq2}$ of length at most
  $\card{X}$ there is no homomorphism $\Completion{\sF_\ww}\to\sX$. Now
  consider the structure $\sX$ with base set $X=V\cup\set{c_1,\dots,c_n}$, and
  relations defined by:
  \begin{itemize}
    \item $P^\sX=V$;

    \item $\iota^\sX=\set{c_1}$, and $\tau^\sX=\set{c_n}$;

    \item $H_a^\sX=\setm{(c_i,c_j)}{j\equiv i+1\mod k}$, for every $a \in \Sigma$;

    \item $S^\sX=\setm{(u,v,c_i,c_j)}{(u,v)\in E,\text{ and }i\not=j}$.
  \end{itemize}
  Then $\sX$ is separated and can be computed from $\sG$ in polynomial time. By the 
  definition  of $\sX$, the set $\{v_1,\dots,v_n\}$ is a clique of size $n$ in
  $\sG$ if and only if $(v_1,\dots,v_n,c_1,\dots,c_n)$ is a $\ww$-code in $\sX$
  for some word $\ww$ of length $n$. Therefore, there is $\ww\in\Sigma^{\geq2}$ and a
  homomorphism $\Completion{\sF_\ww}\to\sX$ if and only if $\sG$ has a clique
  of size $n$. It follows that there is no homomorphism
  $\sX\to\hrushovski{\sT_{\Sigma^{\geq2}}}$ if and only if $\sG$ has a clique of size
  $n$. Hence there is a polynomial-time reduction from the clique problem to
  the complement of $\csp(\hrushovski{\sT_{\Sigma^{\geq2}}})$, and thus
  $\csp(\hrushovski{\sT_{\Sigma^{\geq2}}})$ is \coNP-complete. 
\end{proof}

\subsection{\coNP-intermediate CSPs}
\label{S:Intermediaire}

  Assuming that the complexity classes $\PolComplexity$ and $\coNP$ are
  distinct, we construct a trivial structure such that the CSP of its Hrushovski-encoding is in \coNP, but neither in $\PolComplexity$ nor
  \coNP-complete. The proof is adapted from a construction by Bodirsky and
  Grohe in~\cite{BodirskyGrohe}, which was itself inspired by Ladner's theorem
  on the existence of $\NP$-intermediate problems~\cite{Ladner}. We are, in
  fact, going to prove the following more general result which, similarly to
  Ladner's theorem, implies that there is an infinite hierarchy of such
  \coNP-intermediate $\csp$s.

\begin{theorem} \label{T:ComplexiteIntermediaire}
  Let $L \subseteq \{0,1\}^{\geq 2}$ be a language in $\coNP \setminus
  \PolComplexity$. Then there is a unary language $I \subseteq \{0\}^{\geq 2}$ such
  that $\csp(\hrushovski{\sT_I})$ is also in $\coNP \setminus \PolComplexity$,
  but $L$ is not polynomial-time reducible to $\csp(\hrushovski{\sT_I})$.
\end{theorem}

\begin{proof}
  In this proof, we are going to identify any number  $n \in \N$ with
  the unique word of length $n+2$ in the unary language $\{0\}^{\geq 2}$. Furthermore,
  fix a polynomial-time computable representation of $\theta$-structures as
  binary words $\set{0,1}^{\geq 2}$. For simplicity, assume that each word in
  $\set{0,1}^{\geq 2}$ corresponds to a $\theta$-structure. 

  As in Ladner's proof, we fix an enumeration $M_0,M_1,M_2,\dots$ of all
  deterministic polynomial time Turing machine with input $\set{0,1}^{\geq 2}$ and
  Yes/No output. Moreover, we fix an enumeration of all polynomial time
  reductions, that is, Turing machines $T_0,T_1,T_2,\dots$ with both input and
  output in $\set{0,1}^{\geq 2}$ halting after polynomially bounded time. We can
  assume that both enumerations are computable.

  We are going to construct the set $I = \setm{n \in \N}{ f(n) \text{ is even }
  } \subseteq \N = \{0\}^{\geq 2}$, where $f \colon \N \to \N$ is a function given by a Turing
  machine $F$, which we will define below. The function $f$ is going to be
  non-decreasing and surjective, however, $f$ will grow very slowly. Roughly
  speaking, it will have the property  that for every even $k$
  there is an incremental step from $f(n) = k$ to $f(n+1) = k+1$ if and only if
  we can find a witness $\sX$ such that the Turing machine $M_{k/2}$ does not
  solve $\csp(\hrushovski{\sT_{I}})$ within $n+1$ computational steps. On the
  other hand, for odd $k$, the value will increases to $f(n+1) = k+1$ if and
  only if we find a witness that $T_{\floor{k/2}}$ is not a reduction from $L$ to
  $\csp(\hrushovski{\sT_{I}})$ within $n+1$ computational steps. The two
  properties together with $f$ being surjective imply that there is no
  polynomial-time Turing machine solving $\csp(\hrushovski{\sT_{I}})$, nor a
  polynomial time reduction from $L$ to $\csp(\hrushovski{\sT_{I}})$.

  We define the Turing machine $F$ with input and output both from $\N =
 \{0\}^{\geq 2}$ in the following manner.
  \begin{enumerate}
    \item If $n = 0$, then $F$ outputs $0$.

    \item If $n > 0$, compute the values of $F(i)$ for as many values of
            $i=0,1,2,\dots$ as possible in $n$ Turing steps. Then set $k$ to be
            the last computed value $F(j)$.

    \item If $k$ is even, run the `for all' loop (a) for $n$ many Turing steps.
            If $k$ is odd, run the `for all' loop (b) for $n$ many Turing
            steps. In both cases, if no output is computed within those $n$
            steps, output $F(n) = k$.

    \begin{enumerate}
       \item For every $\theta$-structure $\sX$, simulate $M_{k/2}$ on $\sX$,
               compute whether $\sX$ is separated, and compute whether there is
               a $i \le \min( \card{X}, j)$ such that $F(i)$ is even and
               $\Completion{\sF_i} \to \sX$. Return $k + 1$ if the following
               equivalence holds
               \begin{equation} \label{E:EquivMEchec} 
                  M_{k/2}\text{ accepts $\sX$}\Leftrightarrow \text{$\sX$ is
                  not separated or $\exists i\le \card{X}$ such that
                  $F(i)$ even and $\Completion{\sF_i} \to \sX$}. \tag{$\star$} 
               \end{equation}

       \item For every word $\wu$ in $\set{0,1}^{\geq 2}$, simulate
                $T_{\floor{k/2}}$ on $\wu$ and consider the result
                $T_{\floor{k/2}}(\wu)$ as a $\theta$-structure $\sX$. Compute
                whether $\wu\in L$, compute whether $\sX$ is
                separated, and check whether there is an $i\le \min(\card X, j)$
                such that $F(i)$ is even and $\Completion{\sF_i} \to \sX$. Return $k + 1$ 
                if the following equivalence holds
            \begin{equation} \label{E:EquivTEchec} 
              \wu\in L\Leftrightarrow \text{$\sX$ is not separated or $\exists
              i\le\card{X}$ such that $F(i)$ is even and $\Completion{\sF_i} \to
              \sX$}. \tag{$\dagger$}
            \end{equation}
     \end{enumerate}
\end{enumerate}

  Let $f$ be the function computed by $F$. Note that $f$ is well-defined, since
  in the recursive step (2) at most $n$ Turing steps are executed, thus only values of $F(i)$
  for $i < n$ are needed  for the computation. Analogously, in the loops (a) and (b) only values
  $F(i)$ for $i \leq j$ are used. Clearly $F$ has polynomial runtime,
  since in total at most $2n$ Turing steps are executed to compute $F(n)$.
  Therefore we can decide in polynomial time, whether a given $n \in \N$ is an element of $I := \setm{n \in \N}{f(n) \text{ is even}}$. Hence it follows from
  Theorem~\ref{T:CSPcodeLangage} that $\csp(\hrushovski{\sT_I})$ is in \coNP.
  Note that for every $n$ the value of $f(n+1)$ is either $f(n)$ or $f(n)+1$.
  We claim that, in addition, $f$ is unbounded. For contradiction, assume that
  $f$ has a maximal value $m$.

  First, assume that $m$ is even. Then $I$ is cofinite, as only finitely many
  natural numbers are not mapped to $m$ under $f$. 
  Moreover,  the equivalence \eqref{E:EquivMEchec} does not hold for any $\theta$-structure $\sX$ (otherwise there would be an $n \in \N$ with
  $f(n) = m+1$). Thus $M_{m/2}$ accepts a structure $\sX$ if and only if $\sX$
  is separated and there is no $i \in I$ such that $\Completion{\sF_i} \to
  \sX$. By Lemma~\ref{L:DescriptionCSP} this implies that the polynomial time
  Turing machine $M_{k/2}$ solves $\csp(\hrushovski{\sT_{I}})$. On the other
  hand, by the cofiniteness of $I$ and Theorem~\ref{T:CSPdifficile} we have
  that $\csp(\hrushovski{\sT_{I}})$ is \coNP-complete; a contradiction to our
  assumption that $\PolComplexity\not=\coNP$.

  Next, assume that $m$ is odd. This implies that $I$ is finite, as only
  finitely many elements are not mapped to $m$. Similarly to before, there are
  no word $\wu \in \{0,1\}^{\geq 2}$ such that the equivalence \eqref{E:EquivTEchec}
  holds, that is,  $\wu \in L$ if and only if $T_{\floor{m/2}}(\wu)$ is
  separated and there is no $i \in I$ such that $\Completion{\sF_i} \to
  T_{\floor{k/2}}(\wu)$. Thus $T_{\floor{m/2}}$ is a polynomial-time reduction
  of $L$ to $\csp(\hrushovski{\sT_I})$. Since $I$ is finite, it follows from
  Lemma~\ref{C:EW-Facile} that $\csp(\hrushovski{\sT_I})$ is solvable in
  polynomial time, and hence $L$ is too. This contradicts our assumption $L
  \in \coNP \setminus \PolComplexity$. We conclude that $f$ is
  non-decreasing and surjective. 

  Finally, we show that $\csp(\hrushovski{\sT_I})$ is neither in
  $\PolComplexity$ nor in $\coNP$-complete. To that end, assume that
  $\csp(\hrushovski{\sT_I})$ is solvable in polynomial time. Then there is an
  even integer $k$ such that $M_{k/2}$ solves $\csp(\hrushovski{\sT_I})$. As
  $f$ is surjective, there is an integer $n$ such that $f(n) = k$ and $f(n+1) =
  k+1$. By definition of $f$, there is a $\theta$-structure $\sX$ satisfying
  \eqref{E:EquivMEchec}, that is, witnessing that $M_{k/2}$ does not solve
  $\csp(\hrushovski{\sT_I})$, which is a contradiction. Next, assume that there
  is a polynomial-time reduction from $L$ to $\csp(\hrushovski{\sT_I})$. Then
  there is an odd integer $k$ such that $T_{\floor{k/2}}$ is this reduction. As
  $f$ is surjective, there is an integer $n$ such that $f(n) = k$ and $f(n+1) =
  k+1$. By definition of $f$, there is a word satisfying the equivalence
  \eqref{E:EquivTEchec}. Thus $T_{\floor{k/2}}$ is not a reduction of $L$ to
  $\csp(\hrushovski{\sT_I})$, which is a contradiction.
\end{proof}

 %
%
%

We are now able to summarize the proof of Theorem~\ref{thm:complexity-main}.

\begin{proof}[Proof of Theorem~\ref{thm:complexity-main}]
By Lemma~\ref{L:trivialprop}, the Hrushovski-encoding of any trivial structure is $\omega$-categorical, has slow orbit growth, and satisfies a set of non-trivial h1 identities. For every class $\mathcal C$ that admits $\coNP^{\mathcal C}$-complete problems, we know by Corollary~\ref{cor:CSPcompletness} that there are trivial structures, whose encodings have $\coNP^{\mathcal C}$-complete $\csp$s. By Theorem~\ref{T:ComplexiteIntermediaire} there exists a trivial structure whose encoding has a $\coNP$-intermediate $\csp$ (assuming $\PolComplexity \neq \coNP$).
\end{proof}

We remark that the complexity results in
Theorem~\ref{thm:complexity-main} can be partially replicated for the class of
counterexamples $\sS_{\delta,\alpha}$ from Section~\ref{section:local}.

\begin{theorem}
  Let $\alpha$ and $\delta$ be as in Definition~\ref{definition:S}, and assume
  that $\delta(n)$ can be computed in time polynomial in $n$. Let $L \subseteq
  \{0\}^{\geq 2}$ be an image of $\alpha$ regarded as a unary language. Then the
  following are true.
  \begin{itemize}
  \item $L$ reduces to $\csp(\hrushovski{\sS_{\delta,\alpha}})$ in log-space
          and $\csp(\hrushovski{\sS_{\delta,\alpha}}) \in \coNP^L$.
          
  \item If $L \in \PolComplexity$ then $\csp(\hrushovski{\sS_{\delta,\alpha}})$
          is \coNP-complete.
  \end{itemize}
\end{theorem}

We refrain from giving a proof here; as
for trivial structures, the argument is purely based on the characterization of
$\csp(\hrushovski{\sS_{\delta,\alpha}})$ by Proposition~\ref{prop:emb-forb}.

\subsection{The limited expressive power of identities} \label{section:pseudo}
We can finally prove Theorem~\ref{thm:pseudo}. 
In the following, 
let $\mathcal L$ be the extension of existential second-order logic allowing countably many second-order quantifiers, followed by a countable conjunction of first-order formulas.
It can be seen that the upward direction of \L{}o\'{s}'s theorem and the downward L\"owenheim-Skolem theorem hold for this logic.

\begin{proof}[Proof of Theorem~\ref{thm:pseudo}]
We prove the following: there is no countable set $\Sigma$ of $\theta$-formulas in $\mathcal L$ such that the equivalence $\sA\models\Sigma\Leftrightarrow\csp(\sA)\in\mathcal C$ holds for all homogeneous $\theta$-structures $\sA$.
This proves the theorem, as the satisfaction of a countable set of identities by polymorphisms  can be encoded in $\mathcal L$.
Assume that such a $\Sigma$ exists.
Let $L$ be a language over an alphabet $\Delta$ whose Turing-degree is not intersected by $\mathcal C$, and let $W=\Delta^{\geq2}\setminus L$.
For every $n\in\mathbb N$, let $W\cap \Delta^{\leq n}$ be the set of words of length at most $n$ in $W$. 
Corollary~\ref{cor:homobounded} implies that $\csp(\Hrushovski{\sT_{W\cap \Delta^{\leq n}}})$ can be solved by checking for finitely many forbidden substructures in a given instance, 
therefore $\csp(\Hrushovski{\sT_{W\cap \Delta^{\leq n}}})$ is in AC$^0$.
Since $\Hrushovski{\sT_{W\cap \Delta^{\leq n}}}$ is homogeneous, we get $\Hrushovski{\sT_{W\cap \Delta^{\leq n}}}\models\Sigma$.
Let $\mathcal U$ be a non-principal ultrafilter on $\mathbb N$, and let $\sA$ be the ultraproduct $\left(\prod_{n\in\mathbb N} \Hrushovski{\sT_{W\cap \Delta^{\leq n}}}\right)\kern-3pt/\mathcal U$.
Then $\sA\models\Sigma$ by \L{}o\'{s}'s theorem and $\sA$ is homogeneous, as all the factors in the ultraproduct are homogeneous.
By the L\"owenheim-Skolem theorem, $\sA$ has a countable elementary substructure $\sB$ that also satisfies $\Sigma$.
Note that $\sB$ is homogeneous and has the same age as $\sA$, as it is an elementary substructure of $\sA$.

Finally, we claim that $\sA$ and $\Hrushovski{\sT_W}$ have the same age.
Every finite substructure of $\Hrushovski{\sT_W}$ embeds into $\Hrushovski{\sT_{W\cap \Delta^{\leq n}}}$ for all $n$, by Corollary~\ref{cor:homobounded}, and therefore into their ultraproduct, which is $\sA$.
Conversely, assume that $\sX$ embeds into $\sA$. This precisely means that $I:=\{n\in\mathbb N \mid \sX \text{ embeds into }\Hrushovski{\sT_{W\cap \Delta^{\leq n}}}\}$ is in $\mathcal U$.
Moreover, since $\mathcal U$ is not principal, $I$ is infinite.
Therefore, there is an $n\geq|w|$ such that $\sX$ embeds into $\Hrushovski{\sT_{W\cap \Delta^{\leq n}}}$.
Since $w\in W\cap \Delta^{\leq n}$, Corollary~\ref{cor:homobounded} gives that $\Completion{\sF_w}$ does not homomorphically map to $\sX$, and that $\sX$ is separated.
Since this holds for all $w\in W$, it follows that $\sX$ embeds into $\Hrushovski{\sT_W}$.

By Theorem~\ref{thm:fraisse}, the two structures $\sB$ and $\Hrushovski{\sT_W}$ are isomorphic.
By Theorem~\ref{T:CSPcodeLangage}, $L$ and $\csp(\Hrushovski{\sT_W})$ have the same Turing-degree, therefore $\csp(\Hrushovski{\sT_W})$ is not in $\mathcal C$, a contradiction.
\end{proof}


\section{Discontinuous clone homomorphism to projections}
\label{section:discon}

It was shown in~\cite{BPP-projective-homomorphisms} that there is an $\omega$-categorical structure $\sC$ such that $\Pol(\sC)$ has a discontinuous  clone homomorphism to the projections. This $\sC$ however has an infinite signature, and it can be shown that $\sC$ is not first-order interdefinable with any finite language structure; hence, its polymorphism clone is not finitely related. In this section we use the Hrushovski-encoding to find an $\omega$-categorical finite language structure whose polymorphism clone has a discontinuous clone homomorphism to the projections, proving Theorem~\ref{thm:main:localnoglobal}.

We first recall the construction of $\sC$ in Proposition~4.3 of~\cite{BPP-projective-homomorphisms}. Let
$K$ be the class of all finite structures in the signature $\sigma=(R_n)_{n\geq 1}$, where each $R_n$ names an equivalence relation on injective $n$-tuples with
at most two equivalence classes (seen as a $2n$-ary relation). It is then routine to show that $K$ has the HP and
the SAP, and hence it is a Fra\"{i}ss\'{e} class. Let $\sC'$ be the Fra\"{i}ss\'{e} limit
of $K$. The structure $\sC'$ is $\omega$-categorical since it is homogeneous and since on every finite tuple of elements of its domain, only finitely many of the relations can hold. Now, let $S_n$ be a $3n$-ary relation symbol for every $n\geq 1$.
Let $\sC$ be the expansion of $\sC'$ by relations for these symbols, defined by
\[
  S_n^\sC := \{ (\mathbf{x}, \mathbf{y}, \mathbf{z}) \in (B^{n})^3 \mid \neg(
  R_n^{\sC'}(\mathbf{x}, \mathbf{y}) \land R_n^{\sC'}(\mathbf{y}, \mathbf{z}) )\} 
\]
for all $n \geq 1$.

Since all relations $S_n^\sC$ are definable by quantifier-free first-order formulas over $\sC'$, 
$\sC$ is also $\omega$-categorical, and its age has the SAP. As every polymorphism of $\sC$ preserves $R_n^\sC$, it naturally acts on the two equivalence classes 
of $R_n^\sC$, for every $n\geq 1$. Let $\xi_n$ be the map sending every element of $\Pol(\sC)$ to
its natural action on the equivalence classes of $R_n^\sC$, which we will denoted by
$0$ and $1$ (independently of $n$). Since $f \in \Pol(\sC)$ preserves
$S_n^\sC$, it follows that $\xi_n(f)$ preserves $\{0, 1\}^3 \setminus \{(0,0,0),
(1, 1, 1)\}$. It is a well-known fact~\cite{Post} that such maps are \emph{essentially unary}, i.e., depend on one argument
only. In other words $\xi_n$ is a clone homomorphism from $\Pol(\sC)$ to the function clone of essentially unary functions on $\{0,1\}$, for every $n \geq 1$.

Finally, let $\mathcal{U}$ be a non-principal ultrafilter on the positive integers. We define a map $\xi\colon \Pol(\sC)\to\proj$ by setting, for every $n\geq 1$ and every $n$-ary $f\in\Pol(\sC)$, the value 
$\xi(f)$ to equal the projection $\pi^n_i$ if the set
\[
  D_i(f) := \{n \geq 1 \mid \xi_n(f) \ \text{depends only on the i-th argument}\}
\]
is an element of $\mathcal{U}$. Since $\mathcal U$ is an ultrafilter, this happens for exactly one $1\leq i\leq n$, and thus $\xi$ is well-defined. Then the following results follow from the proof of Proposition~4.3 in~\cite{BPP-projective-homomorphisms}.

\begin{proposition}\label{prop:discontrecap}
  Let $\sC$ and $\xi$ be as defined above. Then
  \begin{enumerate}[label = \textrm{(\alph*)}]
    \item For every set $J$ of positive integers there exists an injective binary function $f_J \in
    \Pol(\sC)$ such that $D_1(f_J) = J$;
		\item $\xi$ is a discontinuous clone homomorphism to $\Projs$.
  \end{enumerate}
\end{proposition}

Roughly speaking, the discontinuity of $\xi$ follows from the fact that we can find a sequence $(J_n)_{n\geq 1}$ of sets of positive integers outside $\mathcal U$ and a converging sequence of binary functions $(f_n)_{n\geq 1}$ with limit $f$ such that $D_1(f_n)=J_n$ for all $n\geq 1$ and such that $D_1(f)\in\mathcal U$. We make this argument more precise in the following proof of Theorem~\ref{thm:main:localnoglobal}, in which we show that the Hrushovski-encoding of $\sC$ also has a discontinuous clone homomorphism to $\Projs$.


\begin{proof}[Proof of Theorem~\ref{thm:main:localnoglobal}]
  First, observe that every relation of $\sC$ is of arity at least $2$, and that
  there are at most $2$ relations of each arity. Hence, if $\Sigma$ is of size
  $2$, there are no more than $|\Sigma|^n$ relations of arity $n$ for all $n
  \geq 2$. Recall that $\sC$  is homogeneous, $\omega$-categorical, and without
  algebraicity.  Our structure with the properties claimed in
  Theorem~\ref{thm:main:localnoglobal} will be $\hrushovski \sC$, the finite
  language encoding of $\sC$ given by Definition~\ref{definition:reducts}. By
  Proposition~\ref{prop:omegacat}, $\EC$ is an $\omega$-categorical structure.
	
  Let $\xi':= \xi \circ \gamma$, where $\gamma$ is the restriction of
  polymorphisms of $\hrushovski \sC$ to $P^{\hrushovski \sC}$ and $\xi$ is as
  in Proposition~\ref{prop:discontrecap}. Recall that by our identification
  convention in Section~\ref{section:finitelang},  $P^{\hrushovski{\sC}} = C$,
  and hence the composition is well-defined. We then claim that $\xi'$ is a
  discontinuous clone homomorphism from $\Pol(\hrushovski \sC)$ to the clone of
  projections. By Proposition~\ref{prop:ucclonehomo}, $\gamma$ is a clone
  homomorphism, thus $\xi'$ is also a clone homomorphism. It only remains to
  show that $\xi'$ is not continuous.

  For every set $J$ of positive natural numbers, let $f_J \in \Pol(\sC)$ be as
  in Proposition~\ref{prop:discontrecap}, and let $g_J \in
  \Pol(\hrushovski{\sC})$ be obtained from it by
  Lemma~\ref{lemma:encodingextensions}~(3) (applied with $\sA = \sB:=\sC$),
  that is, there exists an embedding $u : \coding{\sC}{\sC} \to
  \coding{\sC}{\sC}$ such that $g_J$ extends $u \circ f_J$.  Then, for all
  $n\geq 1$,  $\xi_n(\gamma(g_J))$ only depends on the first argument if
  and only if $\xi_n(f_J)$ also depends on the first argument only. Thus,
  $D_1(\gamma(g_J)) = D_1(f_J)=J$.  Let $J_1 \subseteq J_2 \subseteq \ldots$ be
  a chain of finite subsets of positive natural numbers whose union exhausts
  all such numbers. Then $\xi'(g_{J_i}) = \pi^2_2$ for all $i \geq 1$, as the
  sets $J_i$ are finite and thus not elements of the non-principal ultrafilter
  $\mathcal{U}$.

  Let $\sim$ be the equivalence relation on $\Pol(\hrushovski{\sC})$ given by
  $f \sim g$ if there exists an automorphism $u$ of $\hrushovski{\sC}$ such
  that $f = u \circ g$. By~\cite[Proposition~6]{Topo-Birk} and the fact that
  $\hrushovski{\sC}$ is $\omega$-categorical, we know that
  $\Pol(\hrushovski{\sC})^{(2)}/ \sim$ is a compact space.  Hence, the sequence
  $([g_{J_i}]_\sim)_{i\geq 1}$ has an accumulation point. This means that there
  exist automorphisms $(u_i)_{i\geq 1}$ of $\hrushovski{\sC}$ such that the
  sequence  $(u_i\circ g_{J_i})_{i\geq 1}$ has an accumulation point in
  $\Pol(\EC)$, which we denote by $g$.  Since $\xi'$ is a clone homomorphism,
  we have $\xi'(u_i\circ g_{J_i}) = \xi'(g_{J_{i}}) = \pi^2_2$ for all $i\geq
  1$. 
  
  We now prove that  $\xi_n(\gamma(g))$ depends on its first argument for all
  $n\geq 1$.  Let $n\geq 1$ be arbitrary, and let $k\geq 1$ be such that $n\in
  J_k$ and such that $u_k\circ g_{J_k}$ and $g$ agree on a set containing
  tuples from both equivalence classes of $R_n$; this is possible since $g$ is
  an accumulation point of $(u_i\circ g_{J_i})_{i\geq 1}$.  Since
  $D_1(\gamma(u_k\circ g_{J_k}))=J_k$ and $n\in J_k$, we get that
  $\xi_n(\gamma(u_k\circ g_{J_k}))$ depends on its first argument.  Moreover,
  since $u_k\circ g_{J_k}$ and $g$ agree on a set containing tuples from both
  equivalence classes of $R_n$, it follows that $\xi_n(\gamma(g))$ also depends
  on its first argument, which is what we wanted to show.

  Therefore, we obtain by the definition of $\xi'$ that $\xi'(g)=\xi(\gamma(g))
  = \pi_1^2$.  Thus, $\xi'(g)$ is not an accumulation point of $(\xi'(u_i\circ
  g_{J_i}))_{i\geq 1}$, proving that $\xi'$ is not continuous.
\end{proof}

\bibliographystyle{alpha}
\bibliography{encoding,global}
\end{document}